\documentclass[a4paper,UKenglish,cleveref, autoref, thm-restate,authorcolumns]{lipics-v2019}
\usepackage{amssymb} \usepackage[linesnumbered,lined,boxed,commentsnumbered]{algorithm2e}
\usepackage{wrapfig}
\usepackage{xcolor,pgf,tikz,pgflibraryarrows,pgffor,pgflibrarysnakes}
\usetikzlibrary{fit, shapes} 
\usetikzlibrary{backgrounds} 
\nolinenumbers
\usepackage{blkarray}
\usepgflibrary{shapes}
\usetikzlibrary{snakes,automata}

\tikzstyle{background}=[rectangle,fill=gray!10, inner sep=0.1cm, rounded corners=0mm]

\newcommand{\alphabet}{\Sigma}
\newcommand{\subword}{\preceq}

\newcommand{\upward}{{\uparrow}}
\newcommand{\downward}{{\downarrow}}
\newcommand{\leftend}{{\vdash}}
\newcommand{\rightend}{{\dashv}}
\newcommand\Lang[1]{\mathcal{L}(#1)}

\newcommand{\Cc}{\mathcal{C}}

\newcommand{\Ss}{\mathcal{S}}
\newcommand{\Nn}{\mathbb{N}}

\newcommand{\alp}{\text{alph}}

\newcommand{\mini}{\mathsf{min}}
\newcommand{\config}{\mathsf{config}}

\newcommand{\sem}{\relof}
\newcommand{\OMPA}{\text{OMPA}}

\newcommand{\inn}{\mathsf{in}}
\newcommand{\out}{\mathsf{out}}

\newcommand{\pptl}{\mathsf{PosPTL}}

\newcommand{\ptl}{\mathsf{PTL}}

\newcommand{\tnft}{\mathsf{2NFT}}
\newcommand{\rec}{\mathsf{REG}}

\usepackage{todonotes}
\newcounter{todocounter}

\usepackage{amssymb}
\usepackage{stmaryrd}

\newcommand\varset{{\mathcal X}}

\newcommand\assigned\leftarrow

\newcommand\lbl\lambda

\newcommand{\Ll}{\mathcal{L}}
\newcommand{\Aa}{\mathcal{A}}

\newcommand{\Tt}{\mathcal{T}}

\newcommand{\expspace}{\textsc{Expspace}}
\newcommand{\nexpspace}{\textsc{NExpspace}}
\newcommand{\pspacec}{\textsc{Pspace-Complete}}
\newcommand{\pspaceh}{\textsc{Pspace-Hard}}

\newcommand{\formula}{\Psi}
\newcommand{\constr}{\varphi}

\newcommand{\encode}{{\tt Encode}}

\newcommand{\automaton}{{\mathcal A}}
\newcommand{\term}{t}
\newcommand{\relof}[1]{{\mathcal R}({#1})}

\newcommand{\transducer}{{\mathcal T}}
\newcommand{\langof}[1]{{\mathcal L}({#1})}

\usepackage{listings}
\usepackage{multirow}

\usepackage{amsmath}
\makeatletter
\newcommand\sr[2][]{\ext@arrow 0099{\longrightarrowfill@}{#1}{#2}}
\def\longrightarrowfill@{\arrowfill@\relbar\rightarrow}
\makeatother

\newcommand{\idf}{\mathsf{Id}}

\usepackage{hyperref}

\title{On the Separability Problem of String Constraints}
\author{Parosh Aziz Abdulla}{Uppsala University, Sweden}{parosh@it.uu.se}{}{}
\author{Mohamed Faouzi Atig}{Uppsala University, Sweden}{mohamed\_faouzi.atig@it.uu.se}{}{}
\author{Vrunda Dave}{IIT Bombay, India}{vrunda@cse.iitb.ac.in}{}{}
\author{Shankara Narayanan Krishna}{IIT Bombay, India} {krishnas@cse.iitb.ac.in}{}{}

\authorrunning{P. A.  Abdulla, M. F. Atig, V. Dave, S. Krishna}

\Copyright{Parosh Aziz  Abdulla, M. Faouzi Atig, V. Dave, S. Krishna}


\category{} 

\relatedversion{} 

\supplement{}
\begin{document}

\maketitle

\begin{abstract}
We address the separability problem for straight-line string constraints. The separability problem for languages of a class C by a class S asks: given two languages A and B in C, does there exist a language I in S separating A and B (i.e., I is a superset of A and disjoint from B)?
The separability of string constraints is the same as the fundamental problem of interpolation for string constraints. We first show that regular separability of straight line string constraints is undecidable. Our second result is the decidability of the  separability problem for  straight-line string constraints by  piece-wise testable languages, though the precise complexity is open. In our third result,
we consider the positive fragment of piece-wise testable languages as a separator, and
obtain an EXPSPACE algorithm for the separability of a useful class of straight-line string constraints, and  a PSPACE-hardness result.

\end{abstract}

\section{Introduction}
\label{sec:intro}
 
The \emph{string} data type  is widely used in almost all modern programming and scripting languages.
 Many of the well-known security vulnerabilities such as SQL injections and cross-site scripting 
attacks are often caused by an improper handling of strings.  The detection of such vulnerabilities is usually reduced  to the satisfiability of a formula which is  then solved by {SMT} solvers  (e.g., \cite{DBLP:conf/sp/SaxenaAHMMS10,DBLP:conf/ndss/SaxenaHPS10,DBLP:conf/tacas/YuAB10,DBLP:conf/popl/LinB16}). Therefore,  string constraints solving has received considerable attention in recent years (e.g.  \cite{Chen:2019,DBLP:journals/pacmpl/ChenCHLW18,DBLP:conf/issta/KiezunGGHE09,DBLP:conf/tacas/YuAB10,Zheng13z3str,trinh2014:s3,LiaEtAl-CAV-14,KauslerS14,flatten17,DBLP:conf/popl/LinB16,Berzish2017Z3str3AS,DBLP:journals/pacmpl/HolikJLRV18}) and this has led to the   development of many efficient  string solvers such as HAMPI \cite{DBLP:conf/issta/KiezunGGHE09},
  Z3-str3 \cite{DBLP:conf/fmcad/BerzishGZ17}, 
CVC4~\cite{LiaEtAl-CAV-14,DBLP:journals/fmsd/LiangRTTBD16,DBLP:conf/cav/ReynoldsWBBLT17}, S3P~\cite{trinh2014:s3,Trinh2016},
   Trau \cite{flatten17,trau18,DBLP:conf/atva/AbdullaADHJ19}, SLOTH~\cite{DBLP:journals/pacmpl/HolikJLRV18} and OSTRICH \cite{DBLP:journals/pacmpl/ChenHLRW19}.  %   
   
     In spite  of these advances, most of these tools  do not provide any  completeness guarantees.
  The foundational question regarding  the decidability of string solving
for a large class of string constraints has several challenges to be overcome. A major difficulty  is that any reasonably expressive class of string constraints is either undecidable, or
 has its decidability status open for several years  \cite{DBLP:journals/corr/GaneshB16,DBLP:conf/hvc/GaneshMSR12,DBLP:journals/corr/GaneshMSR13}. In fact, the satisfiability problem is undecidable even for the class of string constraints 
 with concatenation (useful to model assignments in the program) and transduction   (useful to model sanitisation and replacement operations)
 \cite{DBLP:journals/pacmpl/ChenHLRW19}. 	 A direction of research is to  find meaningful and expressive subclasses of string constraints for which the satisfiability problem is decidable (e.g., \cite{DBLP:conf/cav/AbdullaACHRRS14,DBLP:conf/atva/AbdullaADHJ19,DBLP:conf/hvc/GaneshMSR12,DBLP:conf/popl/LinB16,DBLP:journals/pacmpl/HolikJLRV18,DBLP:journals/pacmpl/ChenCHLW18}). 
  An interesting  subclass,  that has been studied extensively, 
  is that of straight-line (SL) string constraints (e.g., \cite{DBLP:journals/pacmpl/HolikJLRV18,DBLP:journals/pacmpl/ChenHLRW19,DBLP:conf/popl/LinB16,DBLP:journals/pacmpl/HolikJLRV18,DBLP:journals/pacmpl/ChenCHLW18}). The SL fragment was introduced by Barcel\'o and Lin in  
  \cite{DBLP:conf/popl/LinB16}. Roughly, an SL constraint models the feasibility of  a path of a string-manipulating program that can be generated by symbolic execution. 
    The \emph{satisfiability} 
  of the SL fragment was shown to be $\mathsf{EXPSPACE}$-complete in 
  \cite{DBLP:conf/popl/LinB16}   and forms the basis of many of the tools above \cite{DBLP:journals/pacmpl/HolikJLRV18,DBLP:journals/pacmpl/ChenCHLW18}.

    In this paper, we focus on the fundamental  problem of 
 \emph{interpolation/separability} for the SL fragment of string constraints. An \emph{interpolant} for a pair of formulas 
$A, B$ is a formula over their common vocabulary that is implied by $A$ and is inconsistent with $B$.  The Craig-Lyndon interpolation technique is very well-known in mathematical logic.     McMillan \cite{DBLP:conf/cav/McMillan03} in his pioneering work, 
has also recognized interpolation as an efficient method for 
automated construction of abstractions of systems. Interpolation based algorithms have been developed for a number of problems in program verification \cite{DBLP:conf/cav/McMillan03,DBLP:conf/tacas/McMillan04,DBLP:conf/cav/McMillan06}. 

Interpolation procedures have been implemented by many  solvers for the theories most commonly used in program verification like linear arithmetic, 
uninterpreted functions with equality and some combination of such theories.  In most of these algorithms, the interpolants were simple.
The interpolation technique can also be used to check the unsatisfiability. In fact, the existence of an interpolant for formulas $A$ and $B$ implies the unsatisfiability of $A \wedge B$.

The notion of \emph{separators} in formal language theory 
is the counterpart of interpolants in logic. 
The separability problem for languages  of a class $\Cc$ by a class $\Ss$ asks: given two languages $I, E \in \Cc$, does there exist a language $S \in \Ss$ separating $I$ and $E$?  That is, $I \subseteq S$ and $S \cap E=\emptyset$. The language $S$ is called the separator of $I, E$. 	
 Separability is a classical problem of fundamental interest in theoretical computer science, and has recently received a lot of attention.  For instance, regular separability has been studied for one-counter automata \cite{slawek19}, Parikh automata \cite{clemente17}, and well-structured transition systems \cite{slawek18}. 
    In the following, we use the terms interpolant or separator of two SL string constraints to mean the same thing, since the solutions of a string constraint can be interpreted as a language.

In this paper, we first show that any string constraint $\phi$ can be written as the conjunction of two 
SL string constraints $A$ and $B$. 
 Therefore, the interpolation problem for the pair $A$ and $B$  can be used to check the   unsatisfiability of  the string constraint $\phi$. (Recall that  the satisfiability problem for general string constraints   is undecidable \cite{DBLP:journals/pacmpl/ChenHLRW19}.) 
 
Then, we consider the regular separability  problem for SL string constraints. We show that this problem is undecidable  (Theorem \ref{thm:reg-sep})  by a reduction from the halting problem of Turing Machines. The main technical difficulty here is to 
ensure that the encoding of a sequence of configurations of a Turing machine results in  SL string constraints.  

Due to this undecidability, we focus on the separability problem of SL string constraints by piece-wise testable languages ($\ptl$). 
 A $\ptl$ 
    is a finite Boolean combination of  special regular languages called \emph{piece languages}  of the form $\Sigma^* a_{1} \Sigma^* a_{2} \dots \Sigma^* a_{n} \Sigma^*$, where all $a_{j} \in \Sigma$. $\ptl$ is  a very natural and well-studied class of languages  in the context of the separability problem (e.g. \cite{place13,cern13,DBLP:journals/dmtcs/CzerwinskiMRZZ17}). Furthermore, 
 among the  various separator classes considered in the literature, the 
 class of piecewise testable languages ($\ptl$) seems to be the most
 tractable:  $\ptl$-separability of regular languages is in $\mathsf{PTIME}$ \cite{place13,cern13}. To decide the $\ptl$-separability of SL string constraints, we first encode the solutions of an SL string constraint as 
the language of an Ordered Multi-Pushdown Automaton (OMPA) (Section \ref{sltoompl}). Then, we show that the $\ptl$-separability of SL  constraints can be reduced to the $\ptl$-separability of OMPAs.
To show the decidability of the latter problem, we first prove that the language of an OMPA: (1)  is a full trio \cite{gins} and  (2) has a semilinear Parikh image. 
Using (1), we obtain the equivalence of the $\ptl$ separability problem and the \emph{diagonal problem} for OMPAs
from \cite{DBLP:journals/dmtcs/CzerwinskiMRZZ17}, where the equivalence has been shown to hold for full trios. 
Next, the decidability of  $\ptl$-separability problem for OMPAs is obtained from the 
decidability of the diagonal problem for OMPAs: the latter is obtained using (2) 
and  \cite{DBLP:journals/dmtcs/CzerwinskiMRZZ17} where the decidability of the diagonal problem has been shown for 
languages having a semilinear Parikh image. As a corollary of these results, we obtain the decidability of the $\ptl$-separability problem  for SL string constraints and OMPAs; however  the exact complexity  is still an open question. In fact, it is   an open problem in the case of  OMPAs with  one stack (i.e., Context-Free Languages (CFLs))  \cite{DBLP:journals/dmtcs/CzerwinskiMRZZ17}.

Given the complexity question, we propose the class 
of positive piecewise testable languages ($\pptl$) as  separators. 
 $\pptl$ is obtained as a negation-free Boolean combination 
of piece languages. As a first result (Theorem \ref{thm:sep-pptl}) we show that deciding $\pptl$-separability 
for any language class has a very elegant proof: it suffices to check if the upward (downward) closure 
 of one of the languages is disjoint from the other language. Using this result, we prove  the  
  $\mathsf{PSPACE}$-completeness of the $\pptl$-separability for CFLs, thereby progressing 
  on the complexity front with respect to a problem which is open in the case of $\ptl$-separability 
   for CFLs. Then, we focus 
   on a class of SL string constraints where the variables used in outputs of the transducers are independent 
   of each other. This class contains SL string constraints with functional transducers (computing 
partial functions, by associating at most one output with each input). 
We prove the decidability and $\mathsf{EXPSPACE}$ membership for the $\pptl$-separability 
of this class by first encoding the solutions of string constraints as outputs of 
 two way transducers ($\tnft$), and then proving the decidability of $\pptl$-separability for $\tnft$.  
 
  Due to lack of space, missing proofs of the technical results  can be found in the appendix.
   
 \smallskip
 
 \noindent
{\bf Related work.}
The satisfiability problem for string constraints is an active research area and there is a lot of progress in the last decade (e.g., \cite{Plandowski:2006:EAS:1132516.1132584,DBLP:conf/issta/KiezunGGHE09,DBLP:conf/popl/LinB16,DBLP:journals/pacmpl/ChenCHLW18,DBLP:journals/pacmpl/ChenHLRW19,DBLP:conf/atva/AbdullaADHJ19,DBLP:conf/hvc/GaneshMSR12,DBLP:journals/corr/GaneshMSR13,DBLP:conf/cav/AbdullaACHRRS14,DBLP:conf/cav/WangTLYJ16}). 
An interpolation based semi-decision procedure for string constraints 
 has been proposed in  \cite{DBLP:conf/cav/AbdullaACHRRS14}. 
As far as we know, this is the first time the separability problem has been studied in the context of string constraints.

 \section{Preliminaries}
 \label{sec:prelims}
 
\smallskip

\noindent
{\bf Notations.}  Let $[i, j]$ denote the set $\{i,  \ldots, j\}$ for $i, j \in \Nn$.
Let $\alphabet$ be a finite alphabet. $\alphabet^*$ denotes the set of all finite words over $\alphabet$ and $\alphabet^+$ denotes $\alphabet^* {\setminus} \{\epsilon\}$ where $\epsilon$ is the empty word. We denote $\Sigma \cup \{\epsilon\}$ by $\Sigma_\epsilon$.
Let $u \in \Sigma^*$. We use $u^R$ to denote the reverse of $u$. The length of the word $u$ is denoted $|u|$ and the $i^{\text{th}}$ symbol of $u$ by  $u[i]$. Given two words $u\in \Sigma^*$ and $v \in \Sigma^*$, we say that $u$ is a subword of $v$ (denoted $u \preceq v$) if there is a mapping $h: [1,|u|] \mapsto [1,|v|]$ such that $(1)$ $u[i]=v[h(i)]$ for all $i \in [1,|u|]$, and $(2)$ $h(i) < h(j)$ for all $i <j$.

\smallskip

\noindent
{\bf (Multi-tape)-Automata.}
A {\it Finite State Automaton} (FSA) over an alphabet $\alphabet$
is a tuple $\automaton=(Q,\alphabet,\delta,I,F)$, where
$Q$ is a finite set of {\it states},
$\delta\subseteq Q\times \Sigma_\epsilon \times Q$
is a set of {\it transitions},  and 
$I\subseteq Q$ (resp. $F\subseteq Q$ ) are the {\it initial} (resp. {\it accepting}) states.
$\automaton$ accepts a word $w$ iff there is a sequence 
$q_0 a_1 q_1 a_2 \cdots a_n q_n$ such that $(q_{i-1}, a_i, q_i)\in \delta$ for all $1\leq i \leq n$, 
$q_0\in I$, $q_n\in F$, and 
$w = a_1\cdot \cdots \cdot a_n$. 
The \emph{language} of $\automaton$, denoted $\langof\automaton$, is the set all accepted words.

Given $n {\in} \mathbb{N}$, a \emph{$n$-tape automaton} $\transducer$  is an automaton over the alphabet $(\Sigma_\epsilon)^n$. It \emph{recognizes} the relation  $\relof{\transducer}{\subseteq} (\alphabet^*)^n$ that contains  the  $n$-tuple of words $(w_1,w_2, \ldots, w_n)$ for which there is a word $(a_{(1,1)},a_{(2,1)},\ldots,a_{(n,1)} )\cdots(a_{(1,m)},a_{(2,m)},\ldots,a_{(n,m)}){\in} \langof\transducer$ with $w_i = a_{(i,1)} \cdot \cdots\cdot a_{(i,m)}$ for all $i \in \{1,\ldots,n\}$.  A \emph{transducer} is  a $2$-tape automaton.

\smallskip

\noindent
{\bf Well-quasi orders.} Given a (possibly infinite set) $C$,  
a quasi-order on $C$ is a reflexive and transitive relation $\sqsubseteq \subseteq C \times C$. 
An infinite sequence $c_1, c_2, \dots$ in $C$ is said to be saturating if there exists 
indices $i < j$ s.t. $c_i \sqsubseteq c_j$. A quasi-order $\sqsubseteq$ is said to be a well-quasi order (wqo) 
on $C$ if every infinite sequence in $C$ is saturating. Observe that the subword ordering $\preceq$ 
between words $u, v$ over a finite alphabet $\Sigma$ 
is well-known to be a wqo  on $\Sigma^*$ \cite{higman}.

\smallskip

\noindent{\em Upward and Downward Closure}. Given a wqo $\sqsubseteq$ on a set $C$, a set $U \subseteq C$ is said to be upward closed if for every $a \in U$ and $b \in C$, with $a \sqsubseteq b$, 
we have $b \in U$. The upward closure of a set $U \subseteq C$ is defined as 
$U {\uparrow}=\{b \in C \mid \exists a \in U, a \sqsubseteq b\}$.
It is known that every upward closed set $U$ can be 
characterized by a finite \emph{minor}. A minor $M \subseteq U$ is s.t. (i) for each $a \in U$, there is a $b \in M$ 
s.t. $b \sqsubseteq a$, and (ii) for all $a, b \in M$ s.t. $a \preceq b$, we have $a=b$.   
For an upward closed set $U$, let $\mini$ be the function that returns the minor of $U$.  
Downward closures are defined analogously. The downward closure of a set $D \subseteq C$ is defined as 
$D {\downarrow}=\{b \in C \mid \exists a \in D, b \sqsubseteq a\}$. The notion of subword relation and thus upward and downward closures naturally extends to $n$-tuples of words. The subword relation here  is component wise i.e. $(u_1, \ldots, u_n) \subword_n (v_1, \ldots, v_n)$ iff $u_i \subword v_i$ for all $i \in [1, n]$.

\smallskip

\noindent
{\bf String Constraints.} An atomic string constraint $\constr$  over an alphabet $\alphabet$ and a set of string variables $\varset$ is either: $(1)$  a \emph{membership constraint}   of the form $x \in \langof\automaton$ where $x \in \varset$  and $\automaton$ is a FSA (i.e., the evaluation of $x$ is in   the language of a FSA $\automaton$ over $\alphabet$),  or $(2)$ 
a \emph{relational constraint}  of the form  $(\term',\term)  \in  \relof{\transducer}$ where  $\term$ and $\term'$ are string terms (i.e., concatenation of   variables in  $\varset$) 
 and $\transducer$ is a transducer over $\alphabet$, and  $\term$ and $\term'$ are  related by a relation recognised by the transducer $\transducer$. 
$(\term',\term)  \in  \relof{\transducer}$ can also be written as $t'=\transducer(t)$, that is,  $\transducer$ produces $t'$ as the output
on input $t$. 
 For a given term $\term$, $|\term|$  denotes the number of variables 
  appearing in $\term$.

A string constraint $\formula$ is a conjunction of atomic string constraints.
We define the semantics of string constraints using a mapping $\eta$, called {\em evaluation}, that assigns for each  variable a word over $\alphabet$. The evaluation $\eta$ can be extended in the straightforward manner to string terms as follows  $\eta(\term_1 \cdot \term_{2}) = \eta(\term_{1}) \cdot \eta(\term_{2})$. We  extend also  $\eta$ to atomic  constraints as follows: (1) $\eta({x \in \langof\automaton})=\top$ iff $ \eta(x)  \in \langof{\automaton}$, and (2)	 $\eta({(\term,\term')  \in  \relof{\transducer}})=\top	$ iff $(\eta(\term),\eta(\term'))  \in  \relof{\transducer} $.

The truth value of  $\formula$ for an evaluation $\eta$ is defined in the standard manner. If $\eta{(\formula)} = \top$ then $\eta$ is a \emph{solution} of $\formula$,  
written $\eta \models \formula$. 
The formula $\formula$ is {\it satisfiable} iff it has a solution.

 A string constraint is said to be {\em Straight Line\footnote{In \cite{DBLP:conf/popl/LinB16}, the authors consider Boolean combinations of membership constraints.  Our results can be extended to handle this. In~\cite{DBLP:conf/popl/LinB16}, they consider also  constraints of the form $x=\term$. Such constraints can be encoded using our relational constraints. }} (SL) 
if it can be  rewritten as $\formula' \wedge \bigwedge\limits_{i=1}^{k} \constr_i$ 
where $\formula'$ is a conjunction of membership constraints, and  $\constr_1,\ldots,\constr_k$ are  relational constraints such that $(1)$ there is a sequence of different  string variables  $x_1, x_2, \ldots, x_n$ with $n \geq k$, and $(2)$ $\constr_i$ is of the form $(x_i,\term_i)  \in  \relof{\transducer_i}$ such that if a variable $x_j$ is appearing in $\term_i$ then $j>i$. A string constraint in the SL form is called an SL formula. Observe that any string formula can be rewritten as a conjunction of two SL formulas (by using extra-variables).

\begin{lemma}
\label{string-splitting}
Given a  string constraint $\formula$, it is possible to construct two SL string constraints $\formula_1$ and $\formula_2$ such that $\formula$ is satisfiable iff $\formula_1 \wedge \formula_2$ is satisfiable.
\end{lemma}

Let $\formula$ be a string constraint and  $x_1, \ldots,x_n$ be the set of variables appearing
in $\formula$. We use $\langof\formula$ to denote the language of $\formula$ which consists of the set of $n$-tuple of words $(u_1, \ldots, u_n) $ such that there is an evaluation $\eta$ with $\eta{(\formula)} = \top$ and $\eta(x_i)=u_i$ for all $i \in [1,n]$.

\smallskip

\noindent
{\bf The Separability Problem.}
Given two  classes of languages $\Cc$ and $\Ss$, the separability problem for $\Cc$ by the separator class $\Ss$ is defined as follows: Given two languages $I$ and $E$ from the class $\Cc$, does there exist a separator $S \in \Ss$ such that $I \subseteq S$ and $E \cap S = \emptyset$.

 \section{Regular Separability of String Constraints}
 \label{sec:regularSL}

Let $\Sigma$ be an alphabet and $k, n$ be two natural numbers. 
A set $R$ of $n$-tuples of words over $\Sigma$ is said to be regular ($\rec$) iff there is a sequence of finite-state automata $\automaton_{(i,1)}, \ldots, \automaton_{(i,n)}$ for every $i \in [1,k]$ such that $ R{=}\bigcup_{i=1}^k [\langof{\automaton_{(i,1)}} \times \cdots \times \langof{\automaton_{(i,n)}}]$. 
The {\em $\rec$} separability problem for string constraints consists in checking for two  given string constraints $\formula$ and $\formula'$ over the string variables $x_1, \ldots,x_n$  whether there is a regular set $R \subseteq (\Sigma^*)^n$ such that $\langof{\formula} \subseteq R$ and $R \cap \langof{\formula'}=\emptyset$. 
The regular separability problem is undecidable in general. 
This can be seen as an immediate corollary of the fact that  the satisfiability problem of string constraints  is undecidable  \cite{morvan,DBLP:journals/pacmpl/ChenHLRW19} even for a simple formula of the form $(x,x)  \in  \relof{\transducer}$ where $\transducer$ is a  transducer and $x$ is a string variable.  
To see why, consider $\formula$ to be $(x,x)  \in  \relof{\transducer}$ and $\formula'$ such that $\langof{\formula'}=\Sigma^*$.  It is easy to see that $\formula'$ and $\formula$ are separable by a regular set iff  $\formula$ is unsatisfiable. 
In the following, we show a stronger result, namely that this undecidability still holds even for $\rec$ separability between two SL formulas. 

\begin{theorem}\label{thm:reg-sep}
The $\rec$ separability problem is undecidable even for  SL  string constraints.  
\end{theorem}

 \section{$\ptl$-Separability of String Constraints}
 \label{sec:ptlSL}
Given the  undecidability of $\rec$ separability, we focus on the separability  problem using  piece-wise testable languages ($\ptl$). 
 We show that the problem is in general  undecidable  and then we show its decidability in the case of SL formulas. The undecidability proof is  exactly the same as in the case of the  $\rec$ separability  (since $\Sigma^*$ is a $\ptl$) while the decidability proof is done by reduction to its corresponding problem for the class of Ordered Multi Pushdown Automata ($\OMPA$) \cite{DBLP:journals/ijfcs/AtigBH17,DBLP:journals/ijfcs/BreveglieriCCC96} (which we show its decidability).  In the rest  of this section,  we first  recall the definition of $\ptl$ and extend it to $n$-tuples of words. Then, we define the class of $\OMPA$s  and show the decidability of its separability problem  by $\ptl$.  Finally, we show the decidability of the separability problem for SL formulas by $\ptl$.

\smallskip

\noindent
{\bf {Piece-wise testable languages.}}\label{subsec:ptl}
Let $\Sigma$ be an alphabet. 
A piece-language is a regular language of the form $\Sigma^* a_1 \Sigma^* a_2 \Sigma^* \ldots \Sigma^* a_k \Sigma^*$ where $a_1 , a_2 , \ldots ,a_k \in \Sigma$.
 The class  of  piecewise testable languages ($\ptl$) is defined as  a finite Boolean combination of piece languages \cite{DBLP:conf/automata/Simon75}. 
We can define $\ptl$ for an $n$-tuple alphabet with $n \in \mathbb{N}$, as follows: The class of  $\ptl$  over $n$-tuple words (denoted $n$-$\ptl$) is defined as the finite Boolean combination of languages of the form  $(\Sigma^*)^n  {\bf v}_1  (\Sigma^*)^n  \cdots   (\Sigma^*)^n   {\bf v}_k  (\Sigma^*)^n$ where  ${\bf v}_i \in (\Sigma_\epsilon)^n$ for all $i \in [1,k]$.

\smallskip

\noindent
{\bf{Ordered Multi Pushdown Automata.}} Let $\Sigma$ be a finite alphabet and $n \geq 1$  a natural number. 
Ordered multi-pushdown automata  extend the model of pushdown automata with multiple stacks. An $n$-Ordered Multi Pushdown Automaton ($\OMPA$ or $n$-$\OMPA$)    is a tuple $\Aa=(Q, \Sigma, \Gamma, \delta, Q_0, F)$  where $(1)$ $Q, Q_0$ and $F$ are finite sets of states, initial states and final states, respectively, $(2)$ $\Gamma$ is the stack alphabet and it contains the special symbol $\bot$, and $(3)$ $\delta$ is the transition relation. 
$\OMPA$ are restricted in a sense that pop operations are only  allowed  from the first non-empty stack. 
A transition in $\delta$ is of the form  $(q, \bot, \ldots,\bot,A_j,\epsilon,\ldots, \epsilon) \to^{a} (q', \gamma_1, \ldots, \gamma_n)$  where $A_j \in \Gamma_\epsilon$ represents the symbol that will be popped from the stack $j$ on reading the input symbol $a \in \Sigma_\epsilon$,  and $\gamma_i \in \Gamma^*$ represents the sequence of symbols which is going to be pushed on the stack $i$.  The condition that $A_{1} = \ldots = A_{j-1} = \bot$
(resp.  $A_{j+1} = \ldots = A_n = \epsilon$) corresponds to the fact  that the stacks $1,\ldots,j-1$ (resp. $j+1, \dots n$) are required to be empty (resp. inaccessible).

A configuration of $\Aa$ is of the form $(q,w,\alpha_1,\ldots,\alpha_n)$ where $q \in Q$, $w \in \Sigma^*$ and $\alpha_1,\ldots,\alpha_n \in (\Gamma \setminus \{\bot\})^* \cdot \{\bot\}$. The transition relation $\to$ between the  set of configurations of $\Aa$ is defined as follows: Given two configurations $(q,w,\alpha_1,\ldots,\alpha_n)$ and $(q',w',\alpha'_1,\ldots,\alpha'_n)$, we have $(q,w,\alpha_1,\ldots,\alpha_n) \to (q',w',\alpha'_1,\ldots,\alpha'_n)$ iff there is a transition $(q, A_1, \ldots, A_n) \to^{a} (q', \gamma_1, \ldots, \gamma_n) \in \delta$ such that $w= a w'$ and $\alpha'_i = \gamma_i u_i$ where $\alpha_i= A_i u_i$ for all $i \in [1,n]$. We use $\to^*$ to denote the transitive and reflexive closure of $\to$. A word $w \in \Sigma^*$ is accepted by $\Aa$ if there exists a sequence of configurations $c_1, \ldots, c_m$ such that: $(1)$ $c_1$ is of the form $(q_0,w, \bot,\ldots, \bot)$, with $q_0 \in Q_0$, $(2)$ $c_m$ is of the form $(q_f,\epsilon, \bot, \ldots, \bot)$, with $q_f \in F$, and $(3)$ $c_i \to c_{i+1}$ for all $i \in [1,m-1]$. The language  of $\Aa$ (denoted by $\Lang{\Aa}$)  is  defined as the set of words accepted by $\Aa$. 
The languages accepted by $\OMPA$ are referred to as OMPL.

In the following, we show that the separability problem for OMPL by $\ptl$ is decidable. As a first step, we show that the class of OMPL forms a 
\emph{full trio} \cite{gins,DBLP:journals/dmtcs/CzerwinskiMRZZ17}. 
   We first recall the definition of a full-trio. Let $L$ be a language 
over an alphabet $A$, and let $B \subseteq A$. The $B$-projection 
of a word $w \in A^*$ is the longest scattered subword containing only symbols   
from $B$. For example, if $A=\{a,b,c\}$, $B=\{b,c\}$, then the $B$-projection 
of $w=ababac$ is $bbc$. The $B$-upward closure of $L$ is the set of all words that can be obtained 
by taking a word in $L$ and padding it with symbols from $B$. For example, if 
$L=\{w\}$ for $w$ as above, then the  $B$-upward closure of $L$ is 
the set $B^*aB^*bB^*aB^*bB^*aB^*cB^*$. 
 A class of languages $\Cc$ is a full trio if it is effectively closed under 
(1) $B$-projection for every finite alphabet $B$,
	(2) $B$-upward closure for every finite alphabet $B$, and 
	(3) intersection with regular languages. 	  

\begin{lemma}
The class of OMPLs forms a full trio. 	
\label{lem:ompl-ft}	
\end{lemma}
 To connect the $\ptl$ separability problem of SL string constraints to that of OMPL, we first use lemma \ref{lem:ptlsepdiag}. 
 Lemma \ref{lem:ptlsepdiag} states that the $\ptl$ separability problem for OMPL is equivalent to the  \emph{diagonal problem} for OMPL. We recall the diagonal problem \cite{DBLP:journals/dmtcs/CzerwinskiMRZZ17}. Fix a class of languages $\Cc$ as above and a language $L \in \Cc$ over alphabet $\Sigma=\{a_1, \dots, a_n\}$. Assume an ordering 
$a_1 < \dots <a_n$ on $\Sigma$.  For $a \in \Sigma$ and $w \in L$, let $\#_a(w)$ denote the number of occurrences of $a$ in $w$. 
The \emph{Parikh image} of $w$ is the $n$-tuple $(\#_{a_1}(w), \dots, 
 \#_{a_n}(w))$. The \emph{Parikh image} of $L$ is the set of all Parikh images of words in $L$. An $n$-tuple $(m_1, \dots, m_n) \in \mathbb{N}^n$ is dominated by another $n$-tuple 
 $(d_1, \dots, d_n)$ iff $m_i \leq d_i$ for all $1 \leq i \leq n$.  The \emph{diagonal problem} for $\Cc$ is the decision problem, which, 
 given as input, a language $L$ from $\Cc$ asks whether 
 each $n$-tuple $(m, \dots, m) \in \mathbb{N}^n$ is dominated by 
 some Parikh image of $L$. 

\begin{lemma}
The $\ptl$-separability and diagonal problems are equivalent  for OMPLs.   	\label{lem:ptlsepdiag}
\end{lemma}

\begin{proof}
This   equivalence has been shown for full trios in  \cite{DBLP:journals/dmtcs/CzerwinskiMRZZ17}  (see Lemma \ref{lem:ompl-ft}).
\end{proof}

\begin{lemma}
Each language $L$ in OMPL has a semilinear Parikh image. 	
\label{lem:omplsemlin}
\end{lemma}
\begin{theorem}
\label{thm-ompl-ptl}
Given two $\OMPA$s $\Aa_1$ and $\Aa_2$, checking whether there is a $\ptl$ $L$ such that $\Lang{\Aa_1} \subseteq L$ and $L \cap \Lang{\Aa_2}  = \emptyset$ is decidable.
\end{theorem}
\begin{proof}
The proof follows from Lemmas \ref{lem:omplsemlin}, \ref{lem:ptlsepdiag} and \cite{DBLP:journals/dmtcs/CzerwinskiMRZZ17}, from where we know that the
 diagonal problem is decidable 
for  classes of languages having   semilinear Parikh images. 
\end{proof}
\begin{remark}
	For the case of $1-$OMPA, the $\ptl$ separability problem is already known to be decidable~\cite{DBLP:journals/dmtcs/CzerwinskiMRZZ17} but its complexity is still an open problem.
\end{remark}

\subsection{From SL formula to OMPA}
\label{sltoompl}
In the following, we show that  the $n$-$\ptl$ separability problem for SL formulas  can be reduced to the $\ptl$ separability problem for OMPLs. To that aim, we proceed as follows: First, we show how to encode  an $n$-tuple of words ($\in (\Sigma^*)^n$) as a word over $(\Sigma \cup \{\#\})^*$. Then, we show how to encode  the set of solutions of  an atomic relational constraint $(x, t) \in \relof{\transducer}$ using the stacks of an OMPA.  Finally, we construct an OMPA that accepts exactly the language of a given SL formula $\formula$. This construction will make use of the constructed OMPAs that encode the set of atomic relational constraints appearing in $\formula$. Let  $\Sigma$  be an alphabet.

\noindent
{\bf Encoding an $n$-tuple of words.}  Let  $n$ be a natural number. We assume w.l.o.g. that the special symbol $\#$ does not belong to $\Sigma$. We define the function $\encode$ that maps any $n$-tuple word ${\bf w}=(w_1,\ldots,w_n) \in (\Sigma^*)^n$ to the word  $w_1 \# w_2 \# \cdots \# w_n$.

\noindent
{\bf From SL atomic relational constraints to OMPAs.}
Let $x_1,x_2, \ldots,x_n$ be a sequence of string variables. Let $P_i$ be a relational constraint of the form $(x_i,\term_i)  \in  \relof{\transducer_i}$ such that if a variable $x_j$ is appearing in the term $\term_i$, then $j>i$.
In the following, we show that we can construct an OMPA $\Aa_i$ with $(3n+|\term_i|+2-3i)$ stacks such that if $\Aa_i$ starts with a configuration where  the first $(n-i)$ stacks contain, respectively,  the evaluations   $\eta(x_{i+1}),\ldots, \eta(x_{n})$ (and all the other stacks are empty), then it can compute an evaluation $\eta(x_i)$ of the variable $x_i$ such that: $(1)$ $ (\eta(x_i),\eta(\term_i))  \in  \relof{\transducer_i}$ and the evaluations  $\eta(x_{i}),\ldots, \eta(x_{n})$ are stored in the last $n-i+1$ stacks of $\Aa_i$. Such an OMPA $\Aa_i$ will be used as a gadget when constructing the OMPA $\Aa$  that  accepts exactly the language of a given SL formula $\formula$. 
\begin{lemma}
\label{lem-OMPA-rel}
We can construct an OMPA $\Aa_i{=}{(Q_i, \Sigma, \{\bot\} \cup \Sigma, \delta_i, \{q_i ^{init}\}, \{q^{final}_i\})}$ with ${(3n+|\term_i|+2-3i)}$-stacks such that for every $u_i,\ldots, u_n \in \Sigma^*$, we have $(q_i^{init}, \epsilon, u_{i+1}\bot,\ldots,u_{n}\bot, \bot, \ldots,\bot) \to^* (q_i^{final},\epsilon, \bot,\ldots,\bot,u_i \bot,u_{i+1}\bot,\ldots,u_n\bot)$ iff $ (\eta(x_i),\eta(\term_i))  \in  \relof{\transducer_i}$ with $\eta(x_j)=u_j$ for all $j \in [i,n]$.
\end{lemma}

\begin{proof}
In the proof, we omit the input $\epsilon$ from the OMPA configurations, 
and only write the state, and stack contents. 
Let us assume that the string term $\term_i$ is of the form $y_{1} y_{2} \cdots y_{|t_i|}$.  
Observe that $y_j \in \{x_{i+1},\ldots,x_n\}$. 
The OMPA $\Aa_i$  proceeds in phases starting from the  configuration  $(q_i^{init},u_{i+1}\bot,\ldots,u_{n}\bot, \bot, \ldots,\bot)$. 
To begin, stacks 1 to $n-i$ contain $u_{i+1}, \dots, u_n$, 
the evaluations of $x_{i+1}, \dots, x_n$, and all other stacks are empty. The computation proceeds in 4 phases. 
The stacks indexed $1, \dots, n-i$ and $n-i+1, \dots, 2n-2i$ 
will be used in the first phase below. The second phase uses stacks indexed $n-i+1, \dots, 2n-2i$ 
and $2n-2i+1, \dots, 2n-2i+|t_i|$ along with the last $n-i$ stacks indexed $2n-2i+|t_i|+3$ to $3n-3i+|t_i|+2$.  
In the third phase, stacks indexed $2n-2i+1, \dots, 2n-2i+|t_i|, 2n-2i+|t_i|+1$ are used. 
In the last phase, stacks indexed 
$2n-2i+|t_i|+1$ and  $2n-2i+|t_i|+2$ are used. At the end of the 4 phases, stacks 
indexed $2n-2i+|t_i|+2,\dots, 3n-3i+|t_i|+2$ hold the evaluations of $x_i, x_{i+1}, \dots, x_n$, and all 
other stacks are empty. 

\vspace{-0.3cm}

\begin{center}
\includegraphics[scale=0.20]{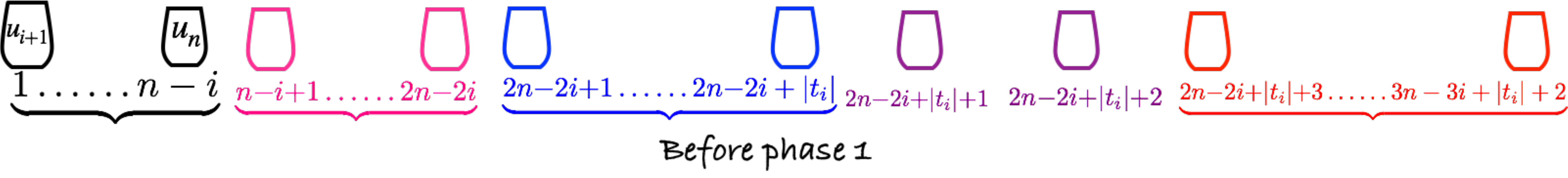}
\end{center}

\vspace{-0.4cm}

\noindent{\bf Phase 1}. The OMPA $\Aa_i$  pops the symbols, one by one, from the first  $(n-i)$-stacks $1, \dots, n-i$ 
and pushes them into the stacks from index $(n-i+1)$ to $(2n-2i)$, respectively. At the end of this phase, the new configuration of the OMPA $\Aa_i$ is
   $(q_i^{init},\bot,\ldots,\bot,u^R_{i+1}\bot,\ldots,u^R_{n}\bot, \bot, \ldots,\bot)$. 
    That is, stacks $n-i+1, \dots, 2n-2i$ have $u_{i+1}^{R}, \dots, u_n^R$, while all other stacks are empty.  
    \vspace{-0.3cm}

   \begin{center}
\includegraphics[scale=0.20]{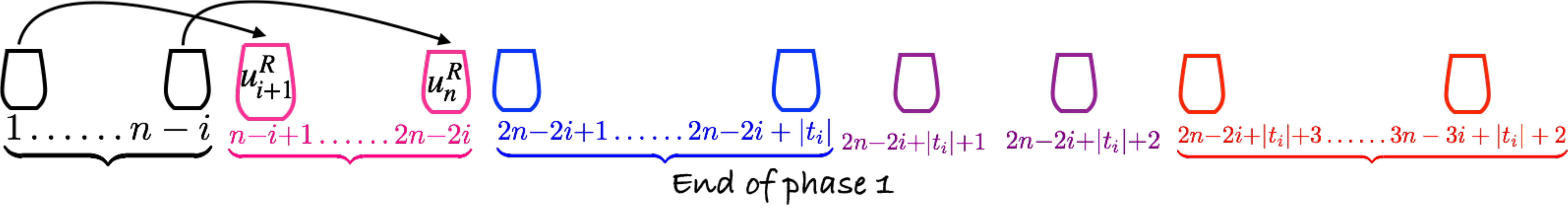}
   \end{center}

\vspace{-0.4cm}

\noindent{\bf Phase 2}. We do two things. (1) the contents of the $n-i$ stacks $n-i+1, \dots, 2n-2i$ are moved (in reverse) into 
the $n-i$ stacks $2n-2i+|t_i|+3, \dots, 3n-3i+|t_i|+2$. This results in the stacks $2n-2i+|t_i|+3, \dots, 3n-3i+|t_i|+2$ containing 
$u_{i+1}, \dots, u_n$. (2) If $y_j$ appearing in $t_i$ is the variable $x_{i+\ell}$, then 
the content of stack $n-i+\ell$ (with $n-i+1 \leq n-i+\ell \leq 2n-2i$) is also 
moved (in reverse) to stack $2n-2i+j$, $1 \leq j \leq |t_i|$.  This results in stack 
$2n-2i+j$ containing $u_{i+\ell}$.
Thus, at the end of (1), (2), the stacks   $n-i+1, \dots, 2n-2i$ are empty, 
 the stack $2n-2i+|t_i|+\ell+2$ contains $u_{i+\ell}$, the evaluation of $x_{i+\ell}$ 
for $\ell \geq 1$, while stack $2n-2i+k$ for $1 \leq k \leq  |t_i|$ contains 
$u_{i+m}$ if $y_k=x_{i+m}$.  The two stacks  $2n-2i+|t_i|+1$ and $2n-2i+|t_i|+2$ are empty at the end of this phase.  
Stack contents of $2n-2i+k$, $1 {\leq} k {\leq} |t_i|$ are referred to as $v_k$ in the figure.

\vspace{-0.3cm}

\begin{center}
\includegraphics[scale=0.20]{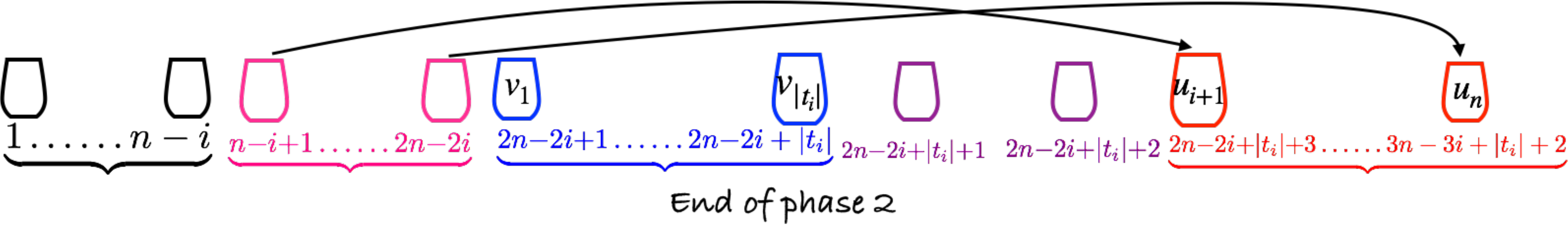}
	
\end{center}
\vspace{-0.4cm}

\noindent{\bf Phase 3}. The OMPA $\Aa_i$  mimics the transducer $\transducer_i$. The current state of $\Aa_i$ is  the same as the current state of the simulated transducer. Each transition of $\transducer_i$ of the form $(q,(a,b),q')$ is simulated by $(1)$ moving the state of $\Aa_i$ from $q$ to $q'$, $(2)$ pushing the symbol $a$ into the stack $(2n-2i+|\term_i|+1)$, and $(3)$ popping the symbol $b$ from the first non-empty stack having  an index between $2n-2i+1$ to  $2n-2i+|\term_i|$. Recall that the stacks $2n-2i+1$ to  $2n-2i+|\term_i|$ contain 
the evaluations of $y_1, \dots, y_{|t_i|}$, for  
$(\eta(x_i), \eta(y_1).\eta(y_2).\dots \eta(y_{|t_i|}) ) {\in} \relof{\transducer_i}$. 
 When the current state of $\Aa_i$ is in a final state of $\transducer_i$ and the stacks from index $2n-2i+1$ to  $2n-2i+|\term_i|$ are empty, 
then we know that $\eta(y_1)\dots \eta(y_{|t_i|})$ is indeed related by $\transducer_i$ on $\eta(x_i)$.  Then  $\Aa_i$ changes its state to $q_{i}^{final}$. 

\vspace{-0.25cm}

\begin{center}
\includegraphics[scale=0.20]{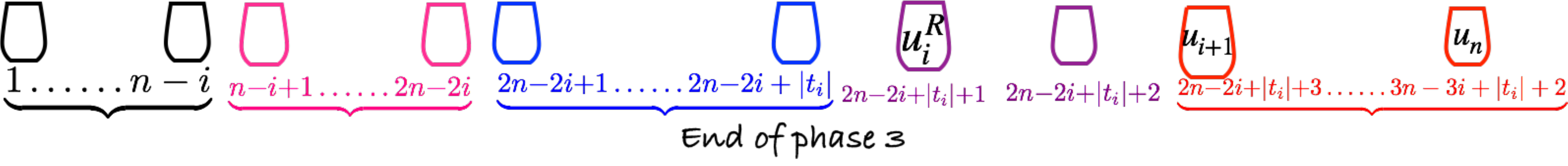}
\end{center}
\vspace{-0.4cm}

\noindent{\bf The Last Phase}. At the end of the third phase, the current configuration of $\Aa_i$ is $(q_i^{final},\bot,\ldots,\bot,u_i^R,\bot, u_{i+1}\bot, \ldots,u_n\bot)$ such that $ (u_i,v_{1}\cdots v_{|t_i|})  \in  \relof{\transducer_i}$: 
that is, the last $n-i$ stacks $2n-2i+|t_i|+3, \dots, 3n-3i+|t_i|+2$ contain 
$u_{i+1}, \dots, u_n$, and stack $2n-2i+|t_i|+1$ contains the reverse of $u_i$. 
Then, $\Aa_i$ pops, one-by-one, the symbols  from the ${(2n-2i+|\term_i|+1)}$-th stack and pushes them, in the reverse order, into the stack $(2n-2i+|\term_i|+2)$. 
Thus, the new configuration of $\Aa_i$ is of the form $(q_i^{final},\bot,\ldots,\bot,u_i\bot, u_{i+1}\bot, \ldots,u_n\bot)$ such that 
$ (u_i,v_{1}\cdots v_{|t_i|})  \in  \relof{\transducer_i}$ where $v_j=u_{\ell}$ if $y_j=x_{\ell}$ for $1 \leq j \leq |t_i|$.
\vspace{-0.3cm}

\begin{center}
\includegraphics[scale=0.20]{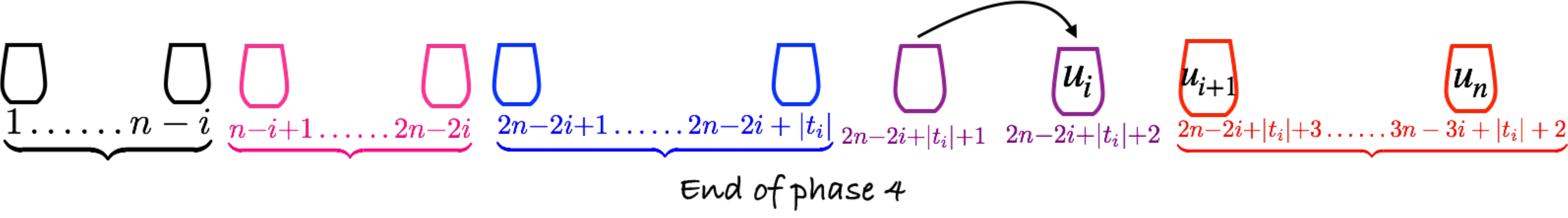}
\end{center}
\vspace{-1.2cm}
\end{proof}

\smallskip

\noindent
{\bf From SL formula  to OMPAs.}
In the following, we first   construct an OMPA that accepts  the encoding of the set of solutions of an SL formula.

\begin{lemma}
Given an SL formula $\formula$, with $x_1, \ldots,x_n$ as its set of variables, it is possible to construct an OMPA $\Aa$ such that $\Lang{\Aa}=\encode{(\langof{\formula})}$.
\label{lem:sl-ompa}
\end{lemma}

\begin{proof}
Let us assume that $\formula$  is of the form $\bigwedge\limits_{i=1}^{n} x_i \in \langof{\automaton_i} \wedge \bigwedge\limits_{i=1}^{k} \constr_i$ 
where   $\constr_1,\ldots,\constr_k$ are  relational constraints such that $\constr_i$ is of the form $(x_i,\term_i)  \in  \relof{\transducer_i}$. The OMPA $\Aa$ will have $(n-k+ \sum_{i=1}^k (2n-2i+2 + |\term_i|))$ stacks.  $\Aa$   first guesses an evaluation  for the variables $x_{k+1}, \ldots, x_{n}$ in the first $n-k$ stacks and then starts simulating the OMPA $\Aa_k$ (see Lemma \ref{lem-OMPA-rel} for the definition of $\Aa_k$) in order to compute a possible evaluation  of the variable $x_k$ such that the relational constraint $(x_k,\term_k)  \in  \relof{\transducer_k}$ holds for that evaluation. 
After this step, the stacks from index $(2n-2k+|\term_k|+2)$ to $(3n-3k+|\term_k|+2)$ contain the evaluation of the string variables $x_k,\ldots,x_n$, and all remaining stacks are empty. 
 Now $\Aa $ can start the simulation of the OMPA  $\Aa_{k-1}$ (Lemma \ref{lem-OMPA-rel}) in order to compute a possible evaluation  of the variable $x_{k-1}$ such that  $(x_k,\term_k)  \in  \relof{\transducer_k} \wedge (x_{k-1},\term_{k-1})  \in  \relof{\transducer_{k-1}}$ holds for that evaluation. At the start of the simulation of $\Aa_{k-1}$ by $\Aa$, the 
$n-k+1$ stacks (indexed $(2n-2k+|\term_k|+2)$ to $(3n-3k+|\term_k|+2)$) contain the evaluations of $x_{k}, \dots, x_n$, and the 
  next $2n-2(k-1)+|\term_{k-1}|+2$ stacks are used to simulate phases 2-4 of  
$\Aa_{k-1}$. At the end of this, the $n-k+2$ stacks backwards from 
 the stack indexed $(3n-3k+|\term_k|+2)+2n-2(k-1)+|\term_{k-1}|+2$ contain the evaluations of $x_{k-1}, \dots, x_n$. Now, 
  $\Aa$ simulates $\Aa_{k-2}, \ldots,\Aa_{n}$ in the same way. 
 At the end of this simulation phase, the last $n$-stacks of $\Aa$ contain an evaluation of the string variables $x_1,\ldots,x_n$ that satisfies $\bigwedge\limits_{i=1}^{k} \constr_i$. Let us assume that the current configuration of $\Aa$ at the end of this  is of the form $(q^{final},\bot,\ldots,\bot,u_1\bot, u_{2}\bot, \ldots,u_n\bot)$. Then, $\Aa$ starts popping, one-by-one, from the $n$-th stack from the last and outputs the read stack symbol $\in \alphabet$ while ensuring that the evaluation $u_1$ of  $x_1$ belongs to $\langof{\automaton_1}$.  
When the $n$-th stack from the last is empty,   $\Aa$ outputs the special symbol $\#$. 
Then,   $\Aa$  does the same  for the $i$-th stack from last, with $i \in [1,n-1]$, which contains the  evaluation of  $x_{i+1}$. 
If $\Aa$ succeeds to empty all  stacks, then this means that the evaluation $\eta$ which associates to the variable $x_i$, the word $u_i$ for all $i \in [1,n]$ satisfies   $\bigwedge\limits_{i=1}^{n} x_i \in \langof{\automaton_i}$. 
Hence,  $u_1\#u_2 \#\cdots \# u_n$ is accepted by $\Aa$ iff $\eta \models \formula$.
\end{proof}

The following lemma shows that the $\ptl$-separability problem for  SL formulas  can be reduced to the $\ptl$-separability problem for OMPLs.  

\begin{lemma}
\label{SL2OMPA}
Let $\formula_1$ and $\formula_2$ be  two SL formulae with $x_1, \ldots,x_n$ as their set of variables. Let   $\Aa_1$ and $\Aa_2$ be two OMPAs such that $\Lang{\Aa_1}{=}\encode{(\langof{\formula_1})}$ and  $\Lang{\Aa_2}{=}\encode{(\langof{\formula_2})}$.  $\formula_1,$  $\formula_2$ are $n$-$\ptl$ separable iff  $\Aa_1,\Aa_2$ are $\ptl$-separable.
\end{lemma}

\noindent As an immediate corollary of Theorem \ref{thm-ompl-ptl}, Lemma \ref{SL2OMPA}, we obtain our main result:

\begin{theorem}
The $n$-$\ptl$ separability problem of SL formulae is decidable.  
\end{theorem}
  
 \section{$\pptl$-Separability of String Constraints}
\label{sec:pptlSL}
In this section, we address the separability problem for string constraints by a sub-class  of $\ptl$, called positive piece-wise testable languages ($\pptl$). A language is in $\pptl$ iff it is defined as a finite positive Boolean combination (i.e., union and intersection but no   complementation) of piece-languages. Given a natural number $n \in \mathbb{N}$, this definition can  naturally be extended to $n$-tuples of words in the straightforward manner (as in the case of $\ptl$)  to obtain the class of  $n$-$\pptl$. 
In the following, we first provide a necessary and sufficient condition for the $n$-$\pptl$ separability problem of  any two languages.

\begin{theorem}\label{thm:sep-pptl}
	Two languages $I$ and $E$  
	are $n$-$\pptl$ separable  
	iff $I {\uparrow} \cap E = \emptyset$ iff $I \cap E \downward = \emptyset$.
	\end{theorem}

The rest of this section is structured as follows: 
First, we  show that the $\pptl$ separability is decidable for OMPLs; in the particular case of CFLs, this problem is \textsc{PSPACE}-complete. 
Then, we use  the encoding of SL formulas to OMPAs (as defined in Section~\ref{sec:ptlSL}), and  show that $n$-$\pptl$ separability of SL formulas reduces to the $\pptl$-separability of corresponding OMPLs. 
Finally, we consider the $\pptl$-separability problem for a subclass of SL formulas, called \emph{right sided} SL formulas.  
We show that the $\pptl$ separability problem for this subclass  is \pspaceh{} and is in \expspace. 

\subsection{$\pptl$ separability of SL formulas}
First, we show that the $\pptl$ separability for OMPLs is decidable.

\begin{theorem}\label{lem:ompa-pptl}
	$\pptl$ separability of OMPLs is decidable.
\end{theorem}
\begin{proof}
Consider $\Cc$ to be the class of OMPLs in Theorem~\ref{thm:sep-pptl}. 
Let $I$ and $E$ be two languages belonging to $\Cc$ as stated in Theorem~\ref{thm:sep-pptl}. Then,  the set  $\mini(I \upward)$ is effectively computable as an immediate consequence of the \emph{Generalized  Valk-Jantzen} construction~\cite{DBLP:journals/corr/AbdullaM13}. 
The main idea behind this construction is to start with an empty minor set $M$ (so to begin, $M \subseteq I$) and keep adding new words $w \in I$ to $M$ if $w$ is not already in $  M\upward$. Before adding a new word, we need to test that  $I \cap \overline{M \upward} \neq \emptyset$ (the complement of 
$M \upward$ intersects with $I$). This test is decidable since (i) OMPLs are closed under intersection with regular languages and 
(ii) the emptiness problem for OMPA is decidable. At each step, we  remove all the non-minimal words from $M$ (since $M$ is finite). The algorithm  terminates due to the Higman's Lemma  \cite{higman} (the minor of an upward closed set is finite). When the algorithm terminates, 
$I \subseteq M \upward$ and thus $I \upward \subseteq M \upward$.  
By construction, $M \subseteq I$ and $M \upward \subseteq I \upward$.  
Thus, $M \upward=I \upward$. Since $M$ is a minor set, we have $\mini(I \upward)=M$. Using (i) and (ii), we obtain the decidability of  checking the emptiness of $I \upward \cap E$, and thus $\pptl$ separability of OMPL is decidable.
%\qed
\end{proof}

As mentioned in section \ref{sec:ptlSL}, the complexity of $\ptl$-separability for 1-OMPL is open; 
however, we show that the $\pptl$ separability  problem for 1-OMPL is \pspacec.

\begin{theorem}
	The $\pptl$-separability for CFLs is \pspacec. 
\label{sep:pda}
\end{theorem}
For the decidability  of the  $n$-$\pptl$ separability of SL formulas,  we use  the encoding of SL formulas to OMPAs (as defined in section~\ref{sec:ptlSL}), and  show that the  $n$-$\pptl$ separability of SL formulas reduces to the $\pptl$ separability of their corresponding OMPLs. The  decidability  of the  $n$-$\pptl$ separability of SL formulas follows from  Theorem \ref{lem:ompa-pptl}.

\begin{lemma}
	\label{lem:sl2ompl-pptl}
	Given two SL formulas $\formula$ and $\formula'$, with $x_1, \ldots,x_n$ as their set of variables. Let   $\Aa$ and $\Aa'$ be two OMPAs such that $\Lang{\Aa}=\encode{(\langof{\formula})}$ and  $\Lang{\Aa'}=\encode{(\langof{\formula'})}$. Then, $\formula$ and $\formula'$ are separable by an $n$-$\pptl$ iff  $\Aa$ and $\Aa'$ are separable by a $\pptl$.
\end{lemma}
As an immediate corollary of Lemma \ref{lem:ompa-pptl} and \ref{lem:sl2ompl-pptl}, we obtain the following theorem:
\begin{theorem}
\label{SL-decidability-thm}
	The $n$-$\pptl$ separability problem of SL formulae is decidable. 
\end{theorem}

\subsection{$\pptl$ separability of right-sided SL formula}\label{subsec:rssl}
Unfortunately, the proof of Theorem  \ref{SL-decidability-thm} does not allow us to extract any complexity result. Therefore, we consider in this subsection a useful fragment of SL formulas, called {\em right-sided} SL formulas. Roughly speaking, an SL formula $\formula$ is {\em right-sided} iff any variable appearing on the right-side of a relational constraint    can not appear on the left-side of any   relational constraint. Let us formalize the notion of right-sided SL formulas.
Let us assume an SL formula $\formula$ of the form $\bigwedge\limits_{i=1}^{n} x_i \in \langof{\automaton_i} \wedge \bigwedge\limits_{i=1}^{k} (x_i,\term_i)  \in  \relof{\transducer_i}$ with   $x_1, \ldots,x_n$ as set of variables. Then,  $\formula$   is said to be {\em right-sided} if  none of the variables $x_1,\ldots,x_k$ appear in  any of $\term_1,\ldots,\term_k$. We call $x_{k+1}, \ldots,x_n$ (resp. $x_1,\ldots,x_k$) {\em independent} (resp. {\em dependent}) variables.
Observe that the  class of SL formulas with functional transducers can be rewritten as  right-sided SL formulas. A  transducer  $\transducer$ is functional if for every word $w$, there is at most one word $w'$ such that $(w',w)  \in  \relof{\transducer}$ ($\transducer$ 
computes a function). An example of a functional transducer is the one implementing the identity relational constraint (allowing to express the equality $x=\term$).

In the following, we show that the $\pptl$-separability problem for right-sided SL formulas is in \expspace. 
To show this result, we will reduce the $\pptl$-separability problem for  right-sided SL formulas to its corresponding problem for  two-way transducers.

\smallskip

\noindent
{\bf Two way transducers.}
Let $\Sigma$ be a finite input alphabet and let $\leftend, \rightend$ be two special symbols not in $\Sigma$.
We assume that every input string $w \in \Sigma^*$ is presented as $\leftend w
\rightend$, where $\leftend, \rightend$ serve as left and right delimiters that
appear nowhere else in $w$.  We write
$\Sigma_{\leftend\rightend}=\Sigma\cup\{\leftend,\rightend\}$.  A two-way
automaton $\Aa=(Q,\Sigma,\delta,I,F)$ has a finite set of states $Q$,
subsets $I,F\subseteq Q$ of initial and final states 
and a transition relation $\delta
\subseteq Q \times \Sigma_{\leftend\rightend} \times Q \times
\{-1,1\}$. The -1 represents that the reading head moves to left after taking the transition  while a 1 represents that it moves to right.

The reading head cannot move left when it is on $\leftend$, 
and cannot move right when it is on $\rightend$. 
A configuration of $\Aa$ on reading $w'= \leftend w \rightend$ is represented by $(q, i)$ where $q \in Q$ and 
$i$ is a position in the input, $ 1 \leq i \leq |w|+2$, which will be read 
in state $q$.  An initial configuration is of the form $(q_0, 1)$ with $q_0 \in I$ and the reading head 
on  $\leftend$. 
If  $w'=w_1 a w_2$ and 
the current configuration is $(q, |w_1|+1)$, and $(q,a,q',-1)\in\delta$, then there is a
transition from the configuration $(q, |w_1|+1)$ to $(q', |w_1|)$  (hence $a\neq\leftend$). 
Likewise, if $(q,a,q',1)\in\delta$, 
we obtain a transition from 
 $(q, |w_1|+1)$ to $(q', |w_1|+2)$.
A run of $\Aa$  on reading $ \leftend w \rightend$ is a sequence of transitions; it is
accepting if it starts in an initial  configuration 
and ends in a configuration of the form $(q, |w|+2)$  with $q \in F$ and the reading head on $\rightend$.
The language of $\Aa$ (denoted $\Lang{\Aa}$)
  is the set of all words
$w\in\Sigma^*$ s.t. $\Aa$ has an accepting run on $\leftend w \rightend$.

We extend the definition of a two-way automaton $\Aa=(Q, \Sigma, \delta, I, F)$ into a two-way transducer ($\tnft$) $\Aa=(Q, \Sigma, \Gamma, \delta, I, F)$ where $\Gamma$ is a finite 
output alphabet.
 The  transition
relation is defined as a \emph{finite} subset $\delta \subseteq Q \times
\Sigma_{\leftend\rightend} \times Q \times \Gamma^* \times \{-1,1\}$.  The
output produced on each transition is appended to the right of the output produced so
far.  $\Aa$ defines a relation $\sem{\Aa}=\{(w,u) \mid w$
is the output produced on an accepting run of $u\}$. The acceptance condition is the same as in two-way automata. 
  Sometimes, we use  the macro-notation $(p,a,q,\alpha,0)$ to denote   a sequence of consecutive  transitions  $(p,a,s,\alpha,d)$ and $(s,b,q,\epsilon,d')$ in $\delta$ with $d+d'=0$, $b \in \Sigma_{\leftend\rightend}$ and $s$ is an extra intermediary  state of $\Aa$ that is not used anywhere else (and that we omit from the set of  states of $\Aa$).

\begin{wrapfigure}[8]{o}{3.1cm}
\vspace{-1cm}

\includegraphics[scale=0.65]{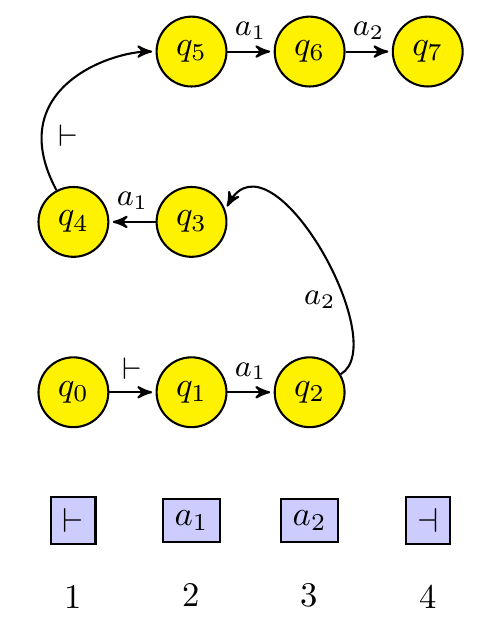}
\end{wrapfigure}

\noindent
{\bf $\pptl$ separability of $\tnft$.}
In the following, we study the $\pptl$ separability of $\tnft$.
We define here the notion of \emph{visiting sequences} (similar to 
\emph{crossing sequences} of $\tnft$ \cite{DBLP:conf/icalp/BaschenisGMP16}), which will be used in the proof of Lemma \ref{lem:min-two-way}.  
Let $w{=}\vdash a_1 \dots a_n \dashv$ be an input word and let $\rho$ be a run of the $\tnft$ on $w$. 
A visiting sequence at a position $x$ in a run $\rho$ of $w$ captures the states visited in order in the run,  
each time the reading head is on position $x$, along with the information 
pertaining to the direction of the outgoing transition from that state. For example, in  run $\rho$, 
if position $x$ is visited for the first time in state $q$, and the outgoing transition chosen in $\rho$ 
from $q$ during that visit had direction +1, then $q^+$ will be the first entry in the visiting sequence.
   For a run $\rho$, the visiting sequence at a position $x$ is defined as the tuple $\rho |x=(q_1^{d_1}, q_2^{d_2}, \dots, q_h^{d_h})$ 
of states that have, in order, visited position $x$  in $\rho$, and whose outgoing transitions had direction $d_1, \dots, d_h$.
 In the example, the visiting sequence at position 2 is $(q_1^+, q_3^-,q_5^+)$, while those at 1 and 3  respectively are  $(q_0^+, q_4^+)$ and  $(q_2^-, q_6^+)$.

\begin{lemma} \label{lem:min-two-way}
	Given a $\tnft$ $\Tt$, if $(v,u)$ is in $\mini(\relof{\Tt}\upward)$ then $|u|$ and $|v|$ are  of at most exponential length in the size of $\Tt$. 	
\end{lemma}

\begin{proof}
		In the following, we show that, if $(v, u) \in \mini(\sem{\Tt}\upward)$, then $|u| \leq   \inn_{\max}= \sum_{i=1}^{|Q|} ((2 |Q|)^i \cdot |\Sigma|)$ and $|v| \leq \out_{\max}=\sum_{i=1}^{|Q|}  ((2 |Q|)^i \cdot |\Sigma| \cdot |Q| \cdot \gamma_{\max})$  where  $\gamma_{\max}$ represents the maximum length of an output on any transition in $\Tt$.
	To show this result, we need to define normalized runs as follows: A run is \emph{normalized}     if it visits each state at most once on each position $x$.
In the following, we show  that $(v,u)$ can be generated by a  normalized run $\rho$. Assume that  $(v,u)$ is accepted by a run  $\rho'$ which  is  not normalized. Then, we will have, in $\rho'$, 
	two visits to some position $x$ of the word in the same state $q$. 
	After the first visit to position $x$ in state $q$, the transducer has explored some positions till its second visit 
	to position $x$ in state $q$. This part does not produce any output since $(v, u)$ is a minimal word. We can delete this 
	explored part of the run  in between, obtaining again, an accepting run, which reads $u$ while producing $v$.  For example, 	in the figure if we have $q_6=q_2$, then we have another run without visiting positions 
	1, 2 for a second time. Observe that repeating this procedure will lead to a normalized run $\rho$ accepting $(v,u)$. 
		The length of visiting sequences in a normalized run is $\leq |Q|$  and hence the number of visiting sequences is at most exponential in $|Q|$,  precisely it is $\leq \sum_{i=1}^{|Q|}{(2 |Q|)}^{i}$.
	Suppose $|u| > \inn_{\max}$. Then there exists a visiting sequence which is repeated on reading the same input symbol in the accepting run of $u$, at positions $i \neq j$.
	By deleting the part between the $i$th and $(j-1)$th position, we again obtain 
	an accepting run over a word $u'$, which is a strict subword of $u$, and whose output $v'$ is also a subword (may not be strict) of $v$, a contradiction to  $(v, u) \in \mini(\sem{\Tt}\upward)$. 
	Now suppose $|u| \leq \inn_{\max}$ but $|v| > \out_{\max}$. 
	We saw that in the normalized run, each visiting sequence  has length at most $|Q|$.
	Then, since $|u| \leq \inn_{\max}$, on reading each position of $u$,  at most $(|Q|) \gamma_{\max}$ symbols can be produced.  Hence, we have $|v| \leq (|Q|) \cdot \gamma_{\max}\cdot  |u|$. 
\end{proof}
From Theorem~\ref{thm:sep-pptl} and Lemma~\ref{lem:min-two-way}, the following result holds:
\begin{lemma}\label{lem:sep-two-way}
	The 2-$\pptl$ separability problem for $\tnft$ is in \expspace.
\end{lemma}
\begin{proof}
	Using Theorem~\ref{thm:sep-pptl}, we know that $\relof{\transducer_1} \upward \cap \relof{\transducer_2} = \emptyset$ iff $\transducer_1$ and $\transducer_2$ are 2-$\pptl$ separable. Here is an \textsc{NExpSpace} algorithm.\\
(1) Guess some $(v, u)$	s.t. the lengths of $v, u$ are at most as given by the proof of Lemma \ref{lem:min-two-way}.  \\
(2) Check if $(v, u) \in 	\relof{\transducer_1}$. If yes, then do (3). Else exit. \\
(3) Check if $(v, u) \uparrow \cap \relof{\transducer_2} \neq \emptyset$. 
The guessed word $(v,u)$ has exponential length in the size  of $\Tt_1$. 
To check if $(v, u)  \in 	\relof{\transducer_1}$, we construct another transducer $\transducer'_1$ that first checks that its input word is $u$, then it comes back to $\vdash$ and starts simulating  
$\transducer_1$, while also keeping track, 
 longer and longer prefixes of  $v$. We then  compare those prefixes with  the output produced by $\transducer_1$. 
 This gives rise to exponentially many states (maintaining prefixes of $u$ and $v$) and we 
 finish when $\transducer_1$ enters an accepting state,  and 
 at the same time, the produced word is $v$. Since $\relof{\transducer'_1}=\{(v, u)\} \cap \relof{\transducer_1}$ by construction, checking  if $(v, u)  \in 	\relof{\transducer_1}$ can be reduced to the emptiness problem of $\transducer'_1$. 
 After this, 
 we check the emptiness of  $(v, u) \uparrow \cap \relof{\transducer_2}$.
  This is done as follows. First, construct automata $\Aa_u, \Aa_v$ accepting 
 languages  $\{u\} \upward$ and $\{v\}\upward$ respectively. The number of states 
 of $\Aa_u, \Aa_v$ are exponential in the number of states of $\transducer_1$, 
 since the lengths of $u, v$ have this bound. 
Then, we construct a transducer $\transducer'_2$ such that $\relof{\transducer'_2}=\{(v, u)\}\uparrow \cap \relof{\transducer_2}$ in a similar manner as $\transducer'_1$. $\transducer'_2$ reads the input word while  simulating $\Aa_u$. 
   On entering an accepting state of  $\Aa_u$, it comes back to $\vdash$. Then 
   it simulates $\Tt_2$, 
   and, on the outputs produced, simulates 
   $\Aa_v$. If $\Aa_v$ enters an accepting state at the same time $\Tt_2$ 
   accepts, then we are done. The state space of $\transducer'_2$ is 
   exponential in the states of $\Tt_1$ and linear in the states of $\Tt_2$.  Since $\relof{\transducer'_2}=\{(v, u)\} \uparrow \cap \relof{\transducer_2}$, checking the emptiness of $(v, u) \uparrow \cap \relof{\transducer_2}$  can be reduced to checking the emptiness problem of $\transducer'_2$.     The emptiness problem for $\tnft$ is known to be  \pspacec{} \cite{DBLP:journals/ipl/Vardi89}. 
   Thus, in our case, the emptiness of $\transducer'_1$ and $\transducer'_2$ can be achieved in space exponential 
   in $\Tt_1$. Since we can handle the second and third steps in exponential space, 
   we obtain an  \nexpspace{} algorithm. By Savitch's Theorem, we obtain the \expspace{} complexity. 
\end{proof}

\noindent
{\bf From Right-sided SL formulas to $\tnft$. }
Hereafter, we show how to encode the set of solutions of a right-sided SL formula using $\tnft$. Let  $\alphabet$ be an alphabet  and  $\# \notin \alphabet$.

\begin{lemma}
Let 	$\formula$  be a right-sided SL formula over  $\alphabet$, with  $x_1,x_2,\ldots,x_n$ as its set of  variables.  Then, it is possible to construct, in polynomial time, a $\tnft$ $\Aa_\formula$ such that $\sem{\Aa_\formula}{=}\{(u_1\#u_2 \# \cdots \#u_n, w_1 \# w_2 \# \ldots  \# w_n) {\mid}  u_1\#u_2 \# \cdots \#u_n {\in} \encode{(\langof{\formula})}$ and  $w_i{=}u_i$ if  $x_i$ is an independent variable $\}$.
	\label{lem:sl-2nft}
\end{lemma}

\begin{proof}
	Let us assume that $\formula$ is of the form $\bigwedge\limits_{i=1}^{n} y_i \in \langof{\automaton_i} \wedge \bigwedge\limits_{i=1}^{k}(y_i,\term_i)  \in  \relof{\transducer_i}$  with  $y_1,\ldots,y_n$ is a permutation of $x_1,\ldots,x_n$. Let $\pi: [1,n] \rightarrow [1,n]$ be the mapping that associates to each index $i \in [1,n]$, the index $j \in [1,n]$ s.t. $x_i=y_j$ (or $x_i=y_{\pi(i)}$). 
	We construct $\Aa_\formula$ as follows:
	$\Aa_\formula$ reads $n$ words over $\Sigma$ separated by $\#$ as input.
	We explain hereafter the working of $\Aa_\formula$ when it produces the assignment for $x_1$ ( the other variables are handled in similar manner).
		
\noindent $\bullet$ Assume that $x_1$ is a dependent variable.
		Let $\constr_{\pi(1)} {=} (y_{\pi(1)}, \term_{\pi(1)}) {\in} \relof{\transducer_{\pi(1)}}$,  with  $\term_{\pi(1)} {=} x_{i_1} x_{i_2} \ldots x_{i_c}$ and $x_{i_j} {\in} \{y_{k+1}, y_{k+2}, \ldots, y_n\}$ for all $j$.
	First, $\Aa_\formula$ reads $x_{i_1}$ i.e. the first variable in $t_{\pi(1)}$. To read $x_{i_1}$, it skips $(i_1 - 1)$ many blocks separated by $\#$s of the input, and comes to $w_{i_1}$.  
		On the first symbol of $w_{i_1}$,  $\Aa_\formula$ starts mimicking transitions of $\transducer_{\pi(1)}$ from its initial state, while producing the same output as $\transducer_{\pi(1)}$.
		On the same output, $\Aa_\formula$ mimics the transitions of $\Aa_{\pi(1)}$ starting from the initial state to check the  membership constraint of $y_{\pi(1)}$.
	This can be done by a product construction between $\Aa_{\pi(1)}$ and $\transducer_{\pi(1)}$. For instance,  $\Aa$ will have a   transition $((p, q), a, (p',q'), b, 1)$  (resp. $((p, q), a, (p',q'), b, 0)$), if there
	 are transitions $(p, (b, a), p')$ (resp. $(p, (b, \epsilon), p')$) in $\transducer_{\pi(1)}$ and $(q, b, q')$ in $\Aa_{\pi(1)}$.
		If it reaches $\#$  or $\rightend$ in the input, it remembers the current states of $\transducer_{\pi(1)}$ and $\Aa_{\pi(1)}$, say $(p_1, q_1)$ in its  control state. 
Next, $\Aa_\formula$ reads $x_{i_2}$ in the input. To read $x_{i_2}$, $\Aa_\formula$ moves to $\leftend$ and then changes direction. 
		As before it reaches $x_{i_2}$ by skipping $(i_2 - 1)$ many $\#$s, and starts reading the input (the first symbol of $w_{i_2}$) from the state $(p_1, q_1)$ stored in the finite control. Transitions are similar to explained above.
		This procedure is repeated to read $x_{i_3}\ldots x_{i_c}$. After  reading $x_{i_c}$, if the next state contains the pair $(p_c, q_c)$, where $p_c$ (resp. $q_c$) is a final state of $\transducer_{\pi(1)}$ (resp. $\Aa_{\pi(1)}$), we can say that the output produced till now satisfies $\constr_{\pi(1)}$ and $y_{\pi(1)} \in \langof{\automaton_{\pi(1)}}$.

		\noindent $\bullet$ Assume now that  $x_1$ is an independent variable, then $\Aa_\formula$ needs to read $x_{1}$.  
		We need a single pass of the input which verifies if  the first
		 block corresponding to value of $x_1$ in input 
		 indeed satisfies its  corresponding  membership constraint. 
		During this pass $\Aa_\formula$  mimics transitions of $\Aa_{\pi(1)}$ starting from its initial states, and outputs  the same letter as input.
		
	The above procedure is repeated for all variables from $x_2$ to $x_n$. After each pass, $\Aa$ moves to $\leftend$ and then changes direction. 
	Irrespective of whether $x_i$ is dependent or not, while going from $x_{i}$ to $x_{i+1}$, $i \in [1, n-1]$, $\Aa$ outputs a $\#$ as separator.	
From the description above, it can be seen that if $x_i$ is independent, then its evaluation $u_i$ 
given as the $i$th block of the input is equal to  the output $w_i$, and if $x_i$ is a dependent variable, 
then the output block  $w_i$ is the output of  $\Tt_{\pi(i)}$. 	\end{proof}
Notice that the above construction of  $\tnft$ relies on the right-sidedness: if  a variable $x_i$  appears in the output of $\Tt_i$ and also 
in the input of $\Tt_k$ for some $k$, then we will have to store  the produced evaluation of $x_i$ in order to use it later on when processing $\Tt_k$. However, there is no way to store the produced evaluation of $x_i$ or compare it with its input evaluation. Next, we show that the $\pptl$ separability problem for  right-sided formulas can be reduced to its corresponding problem for  $\tnft$.

\begin{lemma}\label{lem:pptl-sl-2nft}
	Let  $\formula_1$ and $\formula_2$ be two right-sided SL formula, with $x_1, \ldots, x_n$ as their set of variables. Let   $\Aa_{\formula_1}$ and $\Aa_{\formula_2}$ be the two   $\tnft$s encoding, respectively, the set of solutions of $\formula_1$ and $\formula_2$ (as described in Lemma \ref{lem:sl-2nft}).	Then, the two formulae $\formula_1$ and $\formula_2$ are separable by n-$\pptl$ iff   
	$\relof{\Aa_1}$ and $\relof{\Aa_2}$ are separable by a 2-$\pptl$.
\end{lemma}
\begin{proof}
	Let $\relof{\Aa_{\formula_1}}$ and $\relof{\Aa_{\formula_2}}$ be separable by a 2-$\pptl$ $L$. 
	By definition, $L$ is a Boolean combination (except complementation) of piece languages   of words over the two tuple alphabet $(\Sigma \cup \{\#\})^2$. We can assume w.l.o.g. that $L$ is the union of piece languages. 
	This is possible 	 since the intersection of two  piece languages  can be rewritten as a union of piece languages.
	Consider $L' = L \cap (R \times R)$, where $R$ is a regular language consisting of words having exactly $(n-1)$ $\#$s.
	We claim that $L'$ can be rewritten as the union  of languages of the form $[L_1 \# L_2 \# \ldots \# L_{n}] \times [R_1 \# R_2 \# \ldots \# R_{n}] $ where the $L_i$s and $R_i$s are piece languages over $\Sigma $, and that $L'$ is also a separator of  $\relof{\Aa_{\formula_1}}$ and $\relof{\Aa_{\formula_2}}$.

	We prove this claim inductively.
	As a base  case consider $L$ to be a piece language $((\Sigma \cup \{\#\})^*)^2 (a_1, b_1) ((\Sigma \cup \{\#\})^*)^2 \ldots (a_m, b_m) ((\Sigma \cup \{\#\})^*)^2$. 
	Let $S$ be a finite set containing only the {\em minimal} words of the form $(w, w')$ such that $a_1 a_2 \ldots a_m \preceq w$,  $b_1 b_2 \ldots b_m \preceq w'$, and the symbol $\#$ appears exactly  $(n-1)$-times in  $w$ and $w'$.
	Thus $L \cap (R \times R) = \bigcup\limits_{(a'_1 \ldots a'_k, b'_1 \ldots b'_\ell) \in S} \big[\Sigma ^* a'_1 \Sigma^* a'_2 \ldots a'_k \Sigma^*\big] \times \big[ \Sigma^* b'_1 \Sigma^* b'_2 \ldots b'_\ell \Sigma^*\big]$.
	So $L \cap (R \times R) $ is the union of piece languages of the form $[L_1 \# L_2 \# \ldots \# L_{n}] \times [R_1 \# R_2 \# \ldots \# R_{n}] $ where the $L_i$s and $R_i$s are piece languages over $\Sigma$.
	Now assume that $L$ is of the form  $L_1 \cup L_2$. 
	It is easy to see that $L \cap (R \times R)$ is equivalent to  $(L_1\cap (R \times R)) \cup (L_2 \cap (R \times R))$. 
	Thus we can use our  induction hypothesis to show that $L \cap (R \times R)$ is the union of languages of the form  $[L_1 \# L_2 \# \ldots \# L_{n}] \times [R_1 \# R_2 \# \ldots \# R_{n}] $ where $L_i$ and $R_i$s are piece languages over $\Sigma$.
	Next we prove that $L'$ is a separator of $\relof{\Aa_{\formula_1}}$ and $\relof{\Aa_{\formula_2}}$.
	Indeed if $(v, u) \in \relof{\Aa_{\formula_1}}$, then $(v, u) \in R \times R$, by definition of $\relof{\Aa_{\formula_1}}$. Since $L$ is a separator, we have $(v, u) \in L$ and hence $(v, u) \in L'$.
	Suppose $(v, u) \in \relof{\Aa_{\formula_2}} \cap L'$, then $(v, u) \in L \cap \relof{\Aa_{\formula_1}}$  since $(v, u) \in L'$, and $L' \subseteq L$, which is a contradiction with the assumption that $L$ is a separator.

	Now we are in a condition to provide $n$-$\pptl$ separator for $\langof{\formula_1}$ and $\langof{\formula_2}$, using $L'$.
	Given a language of the form $[L_1 \# L_2 \# \ldots \# L_{n}] \times [R_1 \# R_2 \# \ldots \# R_{n}] $ where $L_i$ and $R_i$s are piece languages, we associate to it an  $n$-$\pptl$ equivalent to $((L_1 \cap R_1) \times (L_2 \cap R_2) \times \ldots \times (L_n \cap R_n))$ : the idea is to 
	generate the $n$ dimensions in the $n$-$\pptl$ from the $n$ $\#$-separated blocks in two dimensions.  
		 This definition is extended in the straightforward manner to union of piece languages. Let $K$ be the $n$-$\pptl$ associated to $L'$.
		$K$ is indeed a separator of $\langof{\formula_1}$ and $\langof{\formula_2}$:
	Suppose ${\bf v }= (w_1, \ldots, w_n) \in \langof{\formula_1}$, then $(w_1 \# \ldots \# w_n, w_1 \# \ldots \# w_n) \in \relof{\Aa_{\formula_1}}$  (from the definition of $\Aa_{\formula_1}$).
	Since $L'$ is a separator, $(w_1 \# \ldots \# w_n, w_1 \# \ldots \# w_n) \in L'$. By construction of $K$, $(w_1, \ldots, w_n) \in K$.
	Assume ${\bf v} = (w_1, \ldots, w_n) \in \langof{\formula_2} \cap K$, then $(w_1  \# \ldots \# w_n, w_1 \# \ldots \#w_n) \in L'$.
	Since $L' \cap \sem{\Aa_{\formula_2}} = \emptyset$,  then $(w_1 \# \ldots \# w_n, w_1 \# \ldots \#w_n) \notin \sem{\Aa_{\formula_2}}$.
	By definition of $\Aa_{\formula_2}$, if $(w_1, \ldots, w_n) \in \langof{\formula_2}$, then $(w_1 \#  \ldots \# w_n, w_1 \# \ldots w_n ) \in \relof{\Aa_{\formula_2}}$. Hence contradiction.

	 For the other direction of the proof, assume the $n$-$\pptl$ $S$ is a separator of $\langof{\formula_1}$ and $\langof{\formula_2}$.
	 Then $S$ can be rewritten as the union  of $(L_1 \times L_2 \times \ldots \times L_n)$ where $L_i$s are piece languages. 
	 Replace each $n$-piece language $(L_1 \times L_2 \times \ldots \times L_n)$ of $S$ with the 2-piece language $(L_1' \# L_2' \# \ldots \# L_n') \times ((\Sigma \cup \{\#\})^* \# \ldots \# (\Sigma \cup \{\#\})^*) $, 
	 where $L_1' =  (\Sigma \cup \{\#\})^* a_1 (\Sigma \cup \{\#\})^* \ldots a_n (\Sigma \cup \{\#\})^*$ if $L_1 = \Sigma^* a_1 \Sigma^* \ldots a_n \Sigma^*$.
	 Denote the union of such languages  by $S'$. It is a 2-$\pptl$ over $(\Sigma \cup \{\#\})$.
	 We show that $S'$ is a 2-$\pptl$ separator of $\relof{\Aa_{\formula_1}}$ and $\relof{\Aa_{\formula_2}}$.
	 Let ${\bf (v, u)} = (v_1 \# \ldots \# v_n, u_1 \# \ldots \# u_n) \in \relof{\Aa_{\formula_1}}$, then $(v_1, \ldots, v_n) \in \langof{\formula_1}$, and thus $(v_1, \ldots, v_n) \in S$ (since $S$ is a separator). This implies that ${\bf (v, u)} \in S'$ by its construction.
	 Suppose ${\bf (v, u)} = (v_1 \# \ldots \# v_n, u_1 \# \ldots \# u_n) \in \relof{\Aa_{\formula_2}} \cap S'$, then $(v_1, v_2, \ldots, v_n) \in \langof{\formula_2}$. Also $(v_1, v_2, \ldots, v_n) \in S$ by construction of $S'$. This leads to a contradiction that $S$ is separator of $\langof{\formula_1}$ and $\langof{\formula_2}$.
	So $S'$ is  a 2-$\pptl$ separator of $\relof{\Aa_{\formula_1}}$ and $\relof{\Aa_{\formula_2}}$.
	\end{proof}
\begin{theorem}
	The $n$-$\pptl$ separability of right-sided SL formulas is  in $\expspace$ and is \pspaceh.
\end{theorem}
\begin{proof}
	Given two right-sided SL formulas $\formula_1$ and $\formula_2$, one can construct corresponding two way transducers $\Aa_{\formula_1}$ and $\Aa_{\formula_2}$ with polynomial states, as mentioned in Lemma~\ref{lem:sl-2nft}. 
	Thanks to Lemma~\ref{lem:pptl-sl-2nft}, the $n$-$\pptl$ separability reduces to 2-$\pptl$ separability of $\relof{\Aa_{\formula_1}}$ and $\relof{\Aa_{\formula_2}}$.
	The 2-$\pptl$ separability of $\tnft$s is  in \expspace{} (Lemma~\ref{lem:sep-two-way}). 
	Hence $n$-$\pptl$ separability of SL formulae is also  in \expspace. 
	For the \pspaceh~ lower bound, we reduce the emptiness of $k$-NFA intersection to $\pptl$ separability of right sided SL. 
	Let $\Aa_1, \ldots, \Aa_k$ be $k$-NFA. We want to decide if  $\bigcap_{i=1}^k  \Aa_i=\emptyset$. 
	Let $\formula_1$ be $\bigwedge_{i=1}^k x_i=x \wedge \bigwedge_{i=1}^k(x_i {\in}  \Aa_i)$, and $\formula_2$ be   $x {\in} \Sigma^*$.
	 $\formula_1$ and $\formula_2$ are $\pptl$ separable iff $\bigcap_{i=1}^k  \Aa_i=\emptyset$. 
\end{proof}

\section{Examples}
We conclude with 2 examples. One, we give a string program whose safety checking boils down to checking 
the separability of two SL string constraints. 
Second, we illustrate how to compute the $\pptl$ separator given 
two SL string constraints, using the $\tnft$ encoding and Theorem~\ref{thm:sep-pptl}. 

\begin{example}

As a practical  motivation of  $\pptl$ (and $\ptl$), consider the following pseudo-PHP code obtained as a variation 
of the code at \cite{php}. In this code,  a   user is prompted to change his password by entering the new password twice. \\
{\small $\texttt{str {\color{blue!60}old} = real\_escape\_string({\color{blue!60}oldIn}); } $\\
$\texttt{str {\color{blue!60}new1} = real\_escape\_string({\color{blue!60}newIn1}); } $\\
$\texttt{str {\color{blue!60}new2 }= real\_escape\_string({\color{blue!60}newIn2}); }$ \\
$\texttt{str {\color{blue!60}pass} = \textcolor{red!60}{database\_query}}(\texttt{"SELECT password FROM users WHERE}$ $\texttt{userID="~{\color{blue!60}userID})};$\\
$\texttt{{\color{red!60}if} (old == pass {\color{red!60}AND} new1 == new2 {\color{red!60}AND} new1 != old )}$ \\
$\texttt{~~~~{\color{red!60}if} (newIn1==newIn2  {\color{red!60}AND} newIn1 != oldIn)}$ \\
$\texttt{~~~~~~~str {\color{green!50!blue}query} = "UPDATE users SET password="~{\color{blue!60}new1} "WHERE userID="~{\color{blue!60}userID};}$  \\
$\texttt{~~~~~~~{\color{red!60}database\_query}({\color{green!50!blue}query});}$ \\
}\\

The user inputs the old password {\tt \color{blue!60}oldIn} and the new password twice : {\tt \color{blue!60}newIn1} and {\tt \color{blue!60}newIn2}. These are sanitized and assigned to 
{\tt \color{blue!60}old}, {\tt \color{blue!60}new1}
and {\tt \color{blue!60}new2} respectively. The old sanitized password is compared with the value {\tt \color{blue!60}pass} 
from the database to authenticate the user, and also 
with the new sanitized password to check that a different password 
has been chosen, and finally, the sanitized new passwords entered twice are checked to be the same. Sanitization ensures that 
there are no SQL injections. To ensure the absence of SQL attacks, we require that the query $\texttt{\textcolor{green!50!blue}{query}}$  does not belong to a regular language   {\tt{Bad}} of bad patterns over some finite alphabet $\Sigma$ (i.e., the program is safe). 
This safety condition can be expressed as the unsatisfiability of the following formula $\varphi$ given by \\
{\small $\mathtt{
new1=T(newIn1)\wedge  new2=T(newIn2) \wedge old= T(oldIn) \wedge  new1=new2 ~\wedge} \mathtt{pass=old \wedge old \neq  new1} $ $\mathtt{\wedge~ newIn1=newIn2 ~\wedge }
\mathtt{\textcolor{green!50!blue}{~query}=u \cdot new1 \cdot v \cdot userID} \wedge \mathtt{{\textcolor{green!50!blue}{query}} \in Bad}$}.

Note that the check $\mathtt{new1=new2}$ has to be done by the server to ensure the sanitized new passwords entered twice are same; however, 
the check $\mathtt{newIn1=newIn2}$ is not redundant, since it can happen that post sanitization, 
the passwords may agree, but not before.  
The sanitization on lines 1, 2 and 3 is represented by the transducer $\tt{T}$  and $\tt{u, v}$ are the
constant strings from line 7.  It is easy to see that the program given here is safe iff the formula $\varphi$  is unsatisfiable.  
Observe that the formula $\varphi$ is not 
in the straight line fragment~\cite{DBLP:conf/popl/LinB16} since variable {\tt new1} has two assignments.  Further, it also has a non-benign chain (see below) 
making it fall out of the fragment of string programs handled in ~\cite{DBLP:conf/atva/AbdullaADHJ19}. However the formula $\varphi$ can be rewritten as a conjunction of the two formula $\varphi_1$ and $\varphi_2$ in straight-line form where \\
   {\small
$\varphi_1 : \mathtt{new1=T(newIn1) \wedge  old= T(oldIn) \wedge   pass=old }\wedge  \mathtt{ \textcolor{green!50!blue}{query}=u \cdot new1 \cdot v \cdot userID} \wedge  \mathtt{\textcolor{green!50!blue}{query} \in Bad}$}  \\
{\small 
$\varphi_2 :  \mathtt{new2=T(newIn2)} \wedge \mathtt{new1=new2} \wedge \mathtt{old \neq new1} \wedge \mathtt{newIn1=newIn2}$}.

It is easy to see that  the program is safe iff  $\varphi_1$ and $\varphi_2$ are separable by the $\pptl$ language that associates to each string variable $\Sigma^*$s.

\subsection*{The string program falls out of the chain-free fragment}

To define the chain-free fragment, \cite{DBLP:conf/atva/AbdullaADHJ19} introduces the notion of a splitting graph.  
Assume we are given a string constraint $\Psi=\bigwedge_{i=1}^n \varphi$ having $n$ relational constraints, 
each of the form $\varphi_j=R_j(t_{2j-1}, t_j)$. Let each term $t_i$ be a concatenation 
of variables $x_{i,1}\dots x_{i,m_i}$. Given such a  string constraint, its splitting graph contains nodes of the 
form $\{(i, j) \mid 1 \leq j \leq 2n, 1 \leq i \leq n_j\}$. Each node $(i, j)$ is labeled by the 
variable $x_{j,i}$.  
The node $(i, 2j-1)$ (resp. $(i, 2j)$) represents the $i$th term in the left side (respectively, the right side) of constraint $\varphi_j$. 
There is an edge from node $p$ to node $q$ if there exists a node $p'$ (different from $q$) such that $p$ and $p'$ represent the nodes corresponding to different sides of the same constraint (say $\varphi_i$) and $p'$ and $q$ have the same label. An edge $(p,q)$, with $p=(i, j)$  is labeled by the $j$th constraint 
$\varphi_j$.

A \emph{chain} in the graph is a sequence of edges of the form $(p_0,p_1)(p_1, p_2) \dots (p_n,p_0)$. 
A chain is a \emph{benign} chain if (1) all relational constraints corresponding to the edges are 
all of the form $R(x, t)$ where $x$ is a single variable ($t$ is a term as usual) and length preserving, and (2)  
the sequence of positions $p_0, p_1, \dots, p_n$ all correspond to 
the left side (or all to the right side).

Recall from $\varphi_1, \varphi_2$ above, that we consider for the string program, the straight-line 
constraints \\
$\phi_1 : \mathtt{new1=T(newIn1)}, \phi_2 :  \mathtt{old= T(oldIn)}, \phi_3 :   \mathtt{pass=old}$, 
$\phi_4 :  \mathtt{new2=T(newIn2)}$,\\
$ \phi_5: \mathtt{new1=new2},  \phi_6 :  \mathtt{old \neq new1}, \phi_7 : \mathtt{newIn1=newIn2}$.

\begin{figure}[h]
	 \tikzstyle{trans}=[-latex, rounded corners]
	\begin{center}
		\scalebox{0.7}{
			
			\begin{tikzpicture}[->,>=stealth',shorten >=1pt,auto, semithick,scale=.9]

			\tikzstyle{every node}=[draw, ellipse, minimum size=1cm, align=center, inner sep=0]

\node[ellipse] at (-6,0) (A1) {$(1, 1)$} ;
\node[ellipse] at (-6,4) (A2) {$(1, 2)$} ;

\node[ellipse] at (-2,0) (B1) {$(1, 7)$} ;
\node[ellipse] at (-2,4) (B2) {$(1, 8)$} ;

\node[ellipse] at (2,0) (C1) {$(1, 9)$} ;
\node[ellipse] at (2,4) (C2) {$(1, 10)$} ;

\node[ellipse] at (6,0) (D1) {$(1, 13)$} ;
\node[ellipse] at (6,4) (D2) {$(1, 14)$} ;

\tikzstyle{every node}=[inner sep=0,outer sep=0]
\node[] (A11) [below=1cm of A1] {{\color{blue}new1}};
\node[] (B11) [below=1cm of B1] {{\color{blue}new2}};
\node[] (C11) [below=1cm of C1] {{\color{blue}new1}};
\node[] (D11) [below=1cm of D1] {{\color{blue}newIn1}};

\node[] (A21) [above=1cm of A2] {{\color{blue}newIn1}};
\node[] (B21) [above=1cm of B2] {{\color{blue}newIn2}};
\node[] (C21) [above=1cm of C2] {{\color{blue}new2}};
\node[] (D21) [above=1cm of D2] {{\color{blue}newIn2}};

       \path(A2) edge  (C1);
\path(A1) edge[bend right=15]  (D1);

\path(B1) edge  (D2);
\path(B2) edge  (C2);

\path(C1) edge  (B1);
\path(C2) edge  (A1);

\path(D1) edge  (B2);
\path(D2) edge[bend right=15]  (A2);

			\end{tikzpicture}
		}
		\end{center}
	\caption{splitting graph\label{fig:split} for the example}
\end{figure}
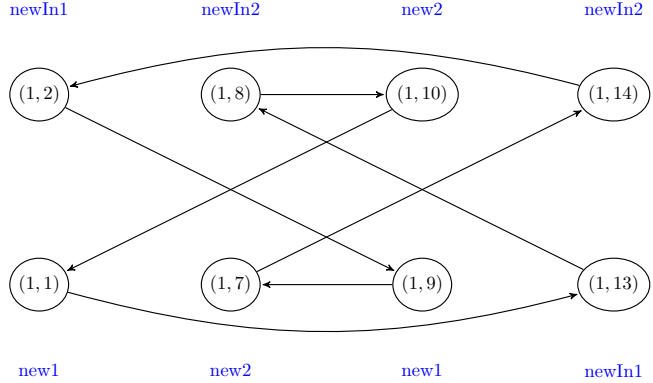

By this definition, the part of the splitting graph which induces a non-benign chain 
is shown in figure \ref{fig:split}:
Observe that there is a chain formed  from the nodes of splitting graph which is:  
$(1, 1) \to (1, 13) \to (1, 8) \to (1, 10) \to (1,1)$. 
This chain is not benign since  the transducer functions are not length preserving.

\end{example}

\begin{example}

We illustrate how to compute the $\pptl$ separator given 
two SL string constraints, using the $\tnft$ encoding and the result from Theorem~\ref{thm:sep-pptl}. 

	Consider two SL string constraints 
	$$\formula_1: (x, ay) \in \relof{\transducer_1} \wedge y \in (e+f)^*, \formula_2: (y, x) \in \relof{\transducer_2} \wedge x \in (e+f)^*$$
	The alphabet is $\Sigma=\{a, e, f\}$. The transducer $\Tt_1$ and $\Tt_2$ implement the following functions
		$\transducer_1(w)=w\upward$ for $w\in a(e+f)^*$, 
		 and $\transducer_2(w')=w'\downward$ for $w'\in (e+f)^*$. 
	
		Note that $\formula_1 \wedge \formula_2$ is not in SL, also it  has a non-benign chain as defined in~\cite{DBLP:conf/atva/AbdullaADHJ19}.
	We show that the languages of formulas $\formula_1$ and $\formula_2$ are separable by a $\pptl$ separator, thereby concluding that $\formula_1 \wedge \formula_2$ is unsatisfiable.
	
	To decide the separability and get this separator, we follow the procedure of Theorem~\ref{thm:sep-pptl}.
	\begin{itemize}
	\item 	 Encode the solutions of formulae to 2NFT.
	\item 
	 $\relof{\Aa_{\formula_1}} = \{(v \# w,u\# w) \mid u \text{~is arbitrary}, w \in (e+f)^*, \text{~and~} v \in \{aw\}\upward\}$. 
	 	 
	 \item $\relof{\Aa_{\formula_2}} = \{(w \# v,w\# u) \mid u \text{~is arbitrary}, w \in (e+f)^*, \text{~and~} v \in \{w\}\downward\}$. 
	 \item Each component in $\relof{\Aa_{\formula_1}}, \relof{\Aa_{\formula_2}}$ is of the form $\eta(x)\#\eta(y)$. 
\item  Decide the emptiness of $\relof{\Aa_{\formula_1}} \upward \cap \relof{\Aa_{\formula_2}}$.  This check reduces to the emptiness of $\mini(\relof{\Aa_{\formula_1}}) \upward \cap \relof{\Aa_{\formula_2}}$. The 
	set of minimal words of $\relof{\Aa_{\formula_1}}\upward$ is $M = \{(a\#, \#)\}$ and $M\upward \cap \relof{\Aa_{\formula_2}}$ is clearly empty. 	\end{itemize}
	Hence $M \upward =((\Sigma \cup \{\#\})^* a (\Sigma \cup \{\#\})^* \# (\Sigma \cup \{\#\})^*)   \times  ((\Sigma \cup \{\#\})^*\# (\Sigma\cup \{\#\})^*)$ is a $\pptl$ separator of 
	$\relof{\Aa_{\formula_1}}, \relof{\Aa_{\formula_2}}$. This gives the $\pptl$ separator $\Sigma^* a \Sigma^* \times \Sigma^*$ 
	for 	$\formula_1,$ $\formula_2$.
		
\end{example}

\bibliography{papers.bib}

\appendix
\newpage
\centerline{\Large{\bf{Appendix}}}

\section{Proof from Section~\ref{sec:prelims} (Proof of lemma~\ref{string-splitting})}
\label{app:string-splitting}
Given a string constraint $\formula$ over variables $x_1, \ldots, x_m$, one can construct two SL string constraints $\formula_1$ and $\formula_2$ s.t. $\formula$ and $\formula_1 \wedge \formula_2$ are equisatisfiable.

By definition $\formula$ is a conjunction of atomic regular membership and relational string constraints i.e. $\formula = \bigwedge_{i=1}^{n} \phi_i \wedge \bigwedge_{i=1}^k \varphi_i$ where $\phi_i$s are membership constraints and $\varphi_i$s are relational constraints. 
Since regular languages are closed under intersection, one can combine the constraints of the form $x_i \in L_1$ and $x_i \in L_2$ as $x_i \in L$ where $L = L_1 \cap L_2$. 
This way we have at most one membership constraint for each variable. 
We can add the conjunction of membership constraints to either $\formula_1$ or $\formula_2$ or both.

Next, we show that by introducing new variables corresponding to each relational constraint $\varphi$, we can partition $\formula$ into two SLs.
For each relational constraint $\varphi: (x_i, t_i) \in \relof{\transducer}$, we add a new variable $u_{\varphi}$ and add a formula $u_{\varphi} = x_i$ (or $(u_{\varphi}, x_i) \in \idf$, where $\idf$ is the identity) to $\formula_1$ and $(u_{\varphi}, t_i) \in \relof{\transducer}$ to $\formula_2$.
After iterating this procedure for all relational constraints $\varphi$, we claim that we have SL formulas in $\formula_1$ as well as $\formula_2$.

Observe that in $\formula_i$ for $i = 1, 2$, variables $x_1, \ldots, x_m$ appear only as part of $(i)$ membership constraint or $(ii)$ input to relational constraint. 
These variables are independent. 
Newly introduced variables $u_{\varphi_j}$ appear as only output of exactly one relational constraint in $\formula_1$ and $\formula_2$.
Hence the variables can be ordered as $u_{\varphi_j}$s having lower precedence than $x_i$s, and both formulas $\formula_1$ and $\formula_2$ preserve the SL syntax. 
Clearly $\formula$ and $\formula_1 \wedge \formula_2$ are equisatisfiable: If we have a solution for $x_i$s in $\formula$, then the same assignment works in $\formula_1 \wedge \formula_2$, in addition $u_{\varphi_j}$s are assigned the values based on values of $x_i$s.
On other side, given a solution for $\formula_1 \wedge \formula_2$, the same assignment of variables $x_i$s works for $\formula$.

\section{Proof from Section~\ref{sec:regularSL} (Proof of theorem~\ref{thm:reg-sep})}
\label{app:reg-sep}
In the following, we show that the $\rec$ separability problem is undecidable even for the subclass SL of  string constraints.
\begin{proof}
	The proof is done by reduction from the halting problem of Turing machines. 
	Consider a deterministic Turing machine $M$ whose set of states is $Q$ and tape alphabet is $\alphabet$, with $a, \#, \$ \notin \alphabet$. 
	Encode a configuration $w$ of $M$ as a word in $\alphabet^* \times (Q \times \alphabet)\times \alphabet^*$ and use $w \rightarrow_M w'$ to denote that $M$ can reach the configuration $w'$ from the configuration $w$ in one step. 
	Let $w_1$ denote the initial configuration of $M$. It can be easily seen that one can design a transducer that accepts the language $\{(w,w') \, |\, w\rightarrow_M w'\}$. 
	This construction can be extended to design two transducers $\transducer_1$ and $\transducer_2$ accepting respectively  the following two languages:
	\medskip

	\noindent
	$L_1= \{(w_1\#w_3\# \cdots\# w_{2k-1}\$ a^j, w_2\#w_4 \#\cdots w_{2k}\# \$ w'_1\# \cdots \#w'_j\#) \mid k {\geq} 1$, 
	$w_{2i-1} \rightarrow_M w_{2i}$  for $i \in [1,k]$ and $w'_1,\ldots,w'_j \in \alphabet^* \times (Q \times \alphabet)\times \alphabet^*$ 
	are arbitrary$
	\}$

	\medskip
	
	\noindent
	$L_2= \{(w_1\#w_3\# \cdots\# w_{2k-1}\$ a^{2j}, w_2\#w_4 \#\cdots w_{2k}\# \$ w'_1\# \cdots \#w'_j\#) \mid 
	k {\geq} 1$,  $w_{2i-2} = w_{2i-1}$ for  $i \in [2,k]$ and 
	$w'_1,\ldots,w'_j \in \alphabet^* \times (Q \times \alphabet)\times \alphabet^*$ are arbitrary$\}$
	
	\medskip
	\noindent Consider the SL formulae $\formula_1{=} (x,y {\$} y) \in \relof{\transducer_1}$ and  $\formula_1{=} (x,y {\$} y) \in \relof{\transducer_2}$. Then, it is easy to see that 
	\medskip

	\noindent
	$\langof{\formula_1}{=} \{(w_1\#w_3\# \cdots\# w_{2k-1}\$ a^k, w_2\#w_4 \#\cdots w_{2k}\#) \mid k {\geq} 1$,  
	$w_2,w_3,\ldots,w_{2k}\in \alphabet^* \times (Q \times \alphabet)\times \alphabet^*$, and $w_{2i-1} \rightarrow_M w_{2i}$ for $i \in [1,k]\}$

	\medskip

	\noindent
	$\langof{\formula_2}{=} \{(w_1\#w_3\# \cdots\# w_{2k-1}\$ a^{2k}, w_2\#w_4 \#\cdots w_{2k}\#) \mid k {\geq} 1$,  
	$w_2,w_3,\ldots,w_{2k}\in \alphabet^* \times (Q \times \alphabet)\times \alphabet^*, w_{2i-2} = w_{2i-1}$ for $i \in [2,k]\}$
	
	\medskip
	We prove that $\formula_1$ and $\formula_2$ are $\rec$ separable if and only if $M$ halts.
	
	\smallskip
	\noindent {\bf{Case 1.}} Suppose $M$ does not halt.
	Assume that there exists a $\rec$ separator $\bigcup_{i=1}^\ell L_i \times L'_i$ separating languages of $\formula_1$ and $\formula_2$. Let 
	$A_i, B_i$ respectively be DFA s.t. $L_i=\langof{A_i}, L'_i=\langof{B_i}$. 
	Let $n{-}1$ be the total number of states in $A_1, \dots, A_{\ell}$.
	Consider the 2-tuple 
	$$(v_1, v_2) = (w_1 \# u_2 \# u_3 \# \ldots \# u_{n!} \$, u_2 \# u_3 \# \ldots \# u_{n!+1} \#  )$$
	where $u_i \rightarrow_M u_{i+1}$ for all $i \in [2, n!]$ and $w_1 \rightarrow_M u_2$.
	
	\noindent Then $(v_1. a^{n!}, v_2) \in \langof{\formula_1}$ and $(v_1. a^{2\cdot (n!)}, v_2) \in \langof{\formula_2}$.
	Suppose that  $(v_1. a^{n!}, v_2) \in \langof{A_i} \times \langof{B_i}$ and that $\delta_{A_i}(q_{i0}, v_1) = p_i$ and $\delta_{A_i}(p_i, a^{n!}) = r_i$ (where $\delta_{A_i}$ denotes the transition function of $A_i$ and $q_{i0}$ is the initial state of $A_i$).
	Then by the pigeon hole principle on the number of states, we must encounter a loop and hence repeat at least one state while reading $a^{n!}$ from $p_i$ in $A_i$. Let  $s_i$  be the period (length of the loop) $1 \leq s_i \leq n$, such that $\delta_{A_i}(p_i, a^{n! + j_is_i}) {=} r_i$ for any $j_i \in \mathbb{N}$.
	In particular, if we choose $j_i = n! \div s_i$, we get $\delta_{A_i}(p_i, a^{2n!}) = r_i$.
	Thus, $(v_1. a^{2n!}, v_2)$ is also  in $\bigcup_{i=1}^\ell L_i \times L'_i$, and which is a contradiction.
	
	\smallskip
	\noindent  \textbf{Case 2.} Suppose $M$ halts in $n > 1$ steps. Then  $(w_1 = u_1) \rightarrow_M u_2 \rightarrow_M \ldots \rightarrow_M u_{n} \nrightarrow_{M} u_{n+1}$.
	We construct the $\rec$ separator for $\langof{\formula_1}$ and $\langof{\formula_2}$ as follows. 
	Let $\config$ represent any configuration of $M$, and let  $R_1{=}(\config \#)^* \config~ \$ a^*$, 
	$R_2{=} (\config \#)^* \config$.
	Then $R_1, R_2$ are regular languages.
	
	We first consider the case when we have $(u, v) \in R_1 \times R_2$ s.t.  number of configurations in both $u, v$  is less than $n$. 
	We construct  $L_i$ which checks whether the number of configurations in $u$ is $i$, 
	and the number of $a$'s in $u$ is also equal to $i$. Likewise, we construct 
	language $L'_i$ which checks if the number of configurations in $v$ is $i$. 
	Then $\bigcup_{i=1}^n L_i \times L'_i$ does not intersect with  $\langof{\formula_2}$ and contains 
	$L(\formula_1)$. 
	Thus for each $i \in [1, n]$,
	$$L_i = \{w_1 \# w_3 \# w_5 \#\ldots \# w_{2i-1} \$ a^i \mid  w_1, w_3, w_5, \ldots, w_{2i-1} \in \Sigma^* \times (Q \times \Sigma) \times \Sigma^*\}$$
	$$L'_i = \{w_2 \# w_4 \# w_6 \#\ldots \# w_{2i}\mid w_2, w_4, w_6, \ldots, w_{2i} \in \Sigma^* \times (Q \times \Sigma) \times \Sigma^* \}$$
	Now consider the case when the number of configurations is more than $n$ in $u$ or $v$, for $(u,v) \in R_1 \times R_2$. 
	It suffices to choose non-deterministically an index $j \leq n$ and check if $w_{2j}$ and $w_{2j+1}$ are not equal and all the configurations before $w_{2j+1}$ and $w_{2j}$ follow the run of $M$. 
	It is sufficient to check up to $n$ configurations since the  run in $M$ has only $n$ configurations in the sequence. For 
	$1 \leq j \leq n$, consider the subset 
	$\mathcal{R}_j$ of $R_1 {\times} R_2$ defined as 
	$\mathcal{R}_j{=} \{u_1 \# u_2 \# \ldots \# u_{j-1} \# w_{2j+1} (\# \config)^* \$ a^*\} \times 
	\{u_2 \# u_3 \# \ldots \# u_{j-1} \# u_{j}  (\# \config)^*\# \}$ 
	such that $u_{j} \neq w_{2j+1}$.  
	$\mathcal{R}_j$  is regular  since the length of configurations $u_j$ and $w_{2j+1}$ are  bounded. $\bigcup_{j=1}^n\mathcal{R}_j \cap L(\formula_2)=\emptyset$.

	Now consider a pair $(u,v) \in R_1 \times R_2$, where  there is no violation as described above in the equality between $u_j$ and $w_{2j+1}$ until $j = n$. Then we cannot continue finding a successor for $u_n$, as a result of which, such pairs  $(u, v) \notin L(\formula_1)$, and hence 
	cannot be considered in separator.

	This way we remove all words in $\langof{\formula_2}$ and accept all the words in $\langof{\formula_1}$. 
	Hence $(\bigcup\limits_{i=1}^{n}L_i \times L'_i) ~\cup~ (\bigcup\limits_{j=1}^{n}\mathcal{R}_j)$ is a $\rec$ separator. 
\end{proof}

\section{Proofs from Section~\ref{sec:ptlSL} }
\subsection{ Proof of Lemma~\ref{lem:ompl-ft}}
\label{app:ompl-full-trio}
	We show that the class of ordered multi pushdown languages (OMPL) form a full trio.
\begin{proof}

	Let $L$ be an OMPL over alphabet 
	$\Sigma$ and let $\Aa$ be the OMPA accepting $L$. We now show the 
	effective closure wrt the three properties. 
	
	\noindent $\bullet$ For $B \subseteq \Sigma$, 
	construct an OMPA from  $\Aa$ by replacing transitions 
	on input symbols $a \in \Sigma \backslash B$ with $\epsilon$.  
	The resulting OMPA accepts the $B$-projection of $L$. For $B \nsubseteq
	\Sigma$, the OMPA accepting the emptyset suffices. 
	
	\noindent $\bullet$  For the $B$-upward closure of $L$, add extra transitions to $\Aa$ as loops on each state on any input symbol from $B$, with no push/pop operations. 
	The resultant OMPA will accept the $B$-upward closure of $L$. 
	
	\noindent $\bullet$  The third property is the closure wrt intersection of regular languages. 
	It is well-known that OMPL are effectively closed wrt  intersection with regular languages. 
	
	This shows that OMPLs form a full trio.
\end{proof}

\subsection{ Proof of Lemma~\ref{lem:omplsemlin}}
\label{app:omplsemlin}
In the following, we present the proof to show that each OMPL $L$ has a semilinear Parikh image. 	
\begin{proof}
	To see this, first note that using the result of \cite{DBLP:journals/ijfcs/AtigBH17}, 
	one can construct a $Dn$-grammar equivalent to the OMPL. (We refer the reader to \cite{DBLP:journals/ijfcs/AtigBH17} for the definition of $Dn$-grammar.)
	Next, using the result from \cite{DBLP:journals/ijfcs/BreveglieriCCC96}, we 
	know that for each $Dn$-grammar $G$ there is an  underlying context-free grammar $G'$ 
	s.t. for each $w \in L(G)$, there is a $w' \in L(G')$ (and conversely) such that  $w$ and $w'$ have the same Parikh image.  Thus, the language generated by the $Dn$-grammar $G$ and the language generated by the context-free grammar $G'$ have the same Parikh image. 
	Since the Parikh image of $L(G')$ is semilinear \cite{salomaa}
	we obtain      that the language of each $Dn$-grammar is semilinear, giving us the result (since $Dn$-grammar are equivalent to the OMPL).
\end{proof}
\subsection{Details for Lemma \ref{lem-OMPA-rel}}

\begin{proof}
	For intuition of this proof, readers can refer the proof provided in Lemma~\ref{lem-OMPA-rel}.
	Formally, set of states of the OMPA $\Aa_i$ is states of $\transducer_i$ in addition with the states $\{q^{init}_i, q^2_{i}, q^{4}_i, q^{final}_i\}$ to distinguish phase 1, 2, and 4. 
	Initial state is defined as $q^{init}_i$.
	
	Phase 1: Being in state $q^{init}_i$, content of first $(n-i)$ stacks are moved to next $(n-i)$ stacks i.e. content of stack $j$, is moved to stack $n-i+j$ for all $1 \leq j\leq n-i$. 
	To achieve this, for all $j$, $1 \leq j \leq n-i$, for all $a \in \Sigma$, we have transition $$(q^{init}_i, \bot((j-1) \text{~times}), a, \epsilon, \ldots, \epsilon) \to^{\epsilon} (q^{init}_i, \epsilon((n-i+j-1) \text{~times}), a, \epsilon, \ldots, \epsilon)$$
	 \begin{center}
		\includegraphics[scale=0.25]{phase1.pdf}
	\end{center}
	
	Phase 2: After phase 1, stacks indexed $1$ to $n-i$ are empty and next $n-i$ stacks contain the valuation of variables $x_{i+1}$ to $x_n$ in reverse. 
	To distinguish this phase, we add a no operation transition from $q^{init}_i$ to $q^{2}_i$. 
	For this, we have a transition $$(q^{init}_i, \bot((n-i) \text{~times}), \epsilon, \ldots, \epsilon) \to^{\epsilon} (q^{2}_i, \epsilon, \ldots, \epsilon)$$
	In this phase, content of non-empty stacks $(n-i+1), \ldots, (2n-2i)$ is moved to last $(n-i)$ stacks indexed $(2n-2i+|\term_i|+2)+1, \ldots, (2n-2i+|\term_i|+2)+(n-i)$. 
	Also, if any of the variable from $x_{i+1}, \ldots, x_n$ is present in $\term_i$ at position $j$ then the content of the corresponding stack is pushed to the stack $(2n-2i)+j$.
	It is done with the following transitions.
	Forall $k$, $1 \leq k \leq (n-i)$, for all $a \in \Sigma$, we have  
	
	$\begin{array}{ll}
	(q^2_i, \bot(n-i+k-1 \text{~times}), a, \epsilon, \ldots, \epsilon) \to^{\epsilon} (q^2_i, &\epsilon ((2n-2i) \text{~times}), \\
	& a \text{~if~} \term_i[1] = x_{i+k}; \text{~else } \epsilon, \\
	& a \text{~if~} \term_i[2] = x_{i+k}; \text{~else } \epsilon, \\
	&\vdots \\
	& a \text{~if~} \term_i[|\term_i|] = x_{i+k}; \text{~else } \epsilon, \\
	& \epsilon((2+k-1)\text{~times}), \\
	& a , \epsilon, \ldots, \epsilon)
	\end{array}$
	 \begin{center}
		\includegraphics[scale=0.25]{phase2.pdf}
	\end{center}

	Phase 3: After phase 2, we have valuation of $\term_i[j]$ present in stack $(2n-2i+j)$. 
	Also, last $(n-i)$ stacks indexed $(2n-2i+|\term_i|+3)$ to $(3n-3i+|\term_i|+2)$ contains the valuation of variables $x_{i+1}, \ldots, x_n$ respectively. 
	In this phase, we evaluate the valuation of $x_i$ based on valuation of $\term_i$, $\eta(\term_i)$. 
	Observe that $\eta(\term_i)$ is stored in the sequence in stacks $(2n-2i+1), \ldots, (2n-2i)+|t_i|$
	and all the stacks before these, are empty.
	Thanks to OMPA, one can pop the content of these stacks one after another, and find the output of these words produced by $\transducer_i$. 
	Output produced is pushed into the stack indexed $(2n-2i+|\term_i|+1)$. 
	To begin this process, we use a no operation transition from $q^2_i$ to $p_0$ where $p_0$ belongs to initial states of $\transducer_i$:
	$$(q^2_i, \bot(2n-2i \text{~times}), \epsilon, \ldots, \epsilon) \to^{\epsilon} (p_0, \epsilon, \ldots, \epsilon)$$
	where $p_0$ is in set of initial state of $\transducer_i$.
	To simulate $\transducer_i$ on content of these $|\term_i|$ stacks, for every transition $(p, (b, a), p')$ of $\transducer_i$, and for all $j$ s.t. $1 \leq j \leq |\term_i|$ we have following corresponding transitions in $\Aa_i$:
	$$(p, \bot(2n-2i+(j-1) \text{~times~}), a, \epsilon, \ldots, \epsilon) \to^\epsilon (p', \epsilon (2n-2i+|\term_i| \text{~times}), b, \epsilon, \ldots, \epsilon)$$
	
	\begin{center}
		\includegraphics[scale=0.25]{phase3.pdf}
	\end{center}
	
	Phase 4: 
	After finishing computing value of $x_i$, if the state of $\transducer_i$ reached is a final state, $\Aa_i$ moves to state $q^4_{i}$ to start the process of this phase.
	$$(p_f, \bot (2n-2i+|\term_i| \text{~times}), \epsilon, \ldots, \epsilon) \to^{\epsilon} (q^4_i, \epsilon, \ldots, \epsilon)$$
	We have already computed the valuation of variable $x_i$ in previous phase.
	But notice that it is stored in stack $(2n-2i+|t_i|+1)$ in reverse order i.e. the top of the stack is the last letter of value of $x_i$, while the bottom represents the first letter of value of $x_i$.
	So there is a next stack indexed $(2n-2i+|\term_i|+2)$ reserved to store its reverse. 
	The operation is carried out using following transition, for all $a \in \Sigma$:
	$$(q^4_i, \bot(2n-2i+|\term_i| \text{~times~}), a, \epsilon, \ldots, \epsilon) \to^\epsilon (q^4_i, \epsilon (2n-2i+|\term_i|+1 \text{~times}), a, \epsilon, \ldots, \epsilon)$$
	This operation is applied until the stack $(2n-2i+|\term_i|+2)$ becomes empty.
	Thus at the end of this phase, we have valuation of $x_{i},\ldots, x_n$ in last $(n-i+1)$ stacks i.e. indexed $(2n-2k+|\term_i|+2), \ldots, (3n-3k+|\term_i|+2)$, and $\Aa_i$ moves to accepting state $q^{final}_i$.
	$$(q^4_i, \bot (2n-2i+|\term_i|+1 \text{~times}), \epsilon, \ldots, \epsilon) \to^{\epsilon} (q^{final}_i, \epsilon, \ldots, \epsilon)$$	
	\begin{center}
		\includegraphics[scale=0.25]{phase4.pdf}
	\end{center}
	
\end{proof}

\subsection{ Proof of Lemma~\ref{SL2OMPA}}

Let $\formula_1$ and $\formula_2$ be  two SL formulae with $x_1, \ldots,x_n$ as their set of variables. Let   $\Aa_1$ and $\Aa_2$ be two OMPAs such that $\Lang{\Aa_1}{=}\encode{(\langof{\formula_1})}$ and  $\Lang{\Aa_2}{=}\encode{(\langof{\formula_2})}$.  $\formula_1,$  $\formula_2$ are $n$-$\ptl$ separable iff  $\Aa_1,\Aa_2$ are $\ptl$-separable.

\begin{proof}
 Lemma \ref{lem:sl-ompa} gives OMPAs $A_{\formula_1}$ and $A_{\formula_2}$ s.t. 
$\Lang{\Aa_{\formula_i}}{=}\encode({\langof{\formula_i}})$. 
Let  
$L$ be a $\ptl$ s.t. $\Lang{\Aa_{\formula_1}} \subseteq L$ and $\Lang{\Aa_{\formula_2}} \cap L=\emptyset$. Both $\Lang{\Aa_{\formula_1}}$ and 
$\Lang{\Aa_{\formula_2}}$ are sets of words of the form $\{w_1 \# w_2 \dots \# w_n \mid w_i \in \Sigma^*\}$, and hence the $\ptl$ $L$ is a finite Boolean combination of languages from $(\Sigma^*\# \Sigma^* )^*$. 
Let $S  \subseteq (\Sigma^* \# \Sigma^* )^*$ be a $\ptl$ consisting  
of all words having exactly $n-1$ $\#$'s.  $S$ is a finite Boolean combination 
of piece languages $\Ll_i$ (assume $1 \leq i \leq p$ for some $p \in \mathbb{N}$)  where each $\Ll_i$ has the form 
$L_{i1} \# L_{i2} \# \dots \#L_{in}$. That is, $L_{i1}, \dots, L_{in}$ are all piece languages over $\Sigma$. It suffices to consider $L=S$.

Given the $\ptl$ $S$ over $\Sigma \cup \{\#\}$, define the 
language $S'$ over $\Sigma^n$, as a finite Boolean combination 
of languages $L_{i1} \times \dots \times L_{in}$ for $1 \leq i \leq p$. 
Since $L_{i1}, \dots, L_{in}$ are piece languages, $S'$ is a $\ptl$ over $\Sigma^n$. 
We show that $S'$ separates $\langof{\formula_1}$ and 
$\langof{\formula_2}$. Consider a word $(w_1, \dots, w_n) \in S'$. 
Then indeed, $w_1\#w_2\#\dots \#w_n \in S$ by definition. 
Since $S$ separates $\Lang{\Aa_{\formula_1}}$ and $\Lang{\Aa_{\formula_2}}$, 
we know that 
$w_1\#w_2\#\dots \#w_n \in \Lang{\Aa_{\formula_1}}$, but not in 
$\Lang{\Aa_{\formula_2}}$. By definition of $\encode{}$, this gives 
$(w_1, \dots, w_n) \in L(\formula_1)$, but not in 
$L(\formula_2)$. Since $S'$ is a $\ptl$ over the $n$-tuple alphabet, 
it is an $n$-$\ptl$ separating $\langof{\formula_1}$ and 
$\langof{\formula_2}$. 

The converse follows on similar lines.
\end{proof}

\section{Proofs from Section~\ref{sec:pptlSL}}

\subsection{Proof of Theorem~\ref{thm:sep-pptl}}\label{app:sep-pptl}
	Two languages $I$ and $E$ from a class $\Cc$ 
	are $n$-$\pptl$ separable  
	iff $I {\uparrow} \cap E = \emptyset$ iff $I \cap E \downward = \emptyset$.
\begin{proof}
	The proof is done for the class of $\pptl$ however it can  easily be extended to languages over $n$-tuples words. 
	
	Let $I, E$ be languages over a finite alphabet $\Sigma$. 
	Assume $I$ and $E$ are separable by the $\pptl$ $S$. 
	Then  $I \subseteq S$ and $S \cap E = \emptyset$. Since $S$ is  upward closed, $I\upward \subseteq S$ and hence 
	$I \upward \cap E = \emptyset$.
	
	Assume $I\upward \cap E = \emptyset$. We claim that $I\upward $ is a $\pptl$-separator. We have already that  
	$I \subseteq I\upward$ and $I \upward \cap E = \emptyset$. 
	Now we will show that $I\upward$ is $\pptl$:  $\mini(I\upward)$ is a finite 
	set since the subword ordering $\preceq$  is a wqo 
	on $\Sigma^*$. Let $\mini(I\upward)=\{w_1, \dots, w_k\}$. Then $I\upward $ can be defined as a finite union of $\{w_i\} \upward$ for $i \in \{1,2, \ldots, k\}$. For $w_i=a_{i_1} \dots a_{i_{m_i}}$, with $a_{i_j} \in \Sigma$,  
	$\{w_i\} \upward$ is the piece language $\Sigma^*a_{i_1}\Sigma^* \dots \Sigma^*a_{i_{m_i}}\Sigma^*$. Thus, 
	$I\upward $ is indeed in $\pptl$, being the finite union of piece languages.  		
	
	To see $I {\uparrow} \cap E = \emptyset$ iff $I \cap E \downward = \emptyset$. Let $w \in I {\uparrow} \cap E$. Let $w' \preceq w$ 
	be s.t. $w' \in I$. Then $w' \in I \cap E \downward$. The converse works similarly.%\qed
\end{proof}

\subsection{$\pptl$ separability of PDA (Proof of Theorem~\ref{sep:pda}) }
\label{app:PDA-PPTL}
	We show that $\pptl$ separability of  context free languages is \pspacec{}. 
	Context free languages are equivalently represented by pushdown automata(PDA) or context free grammars(CFG) and the conversion from CFG to PDA and PDA to CFG is in polynomial time. 
	We recall the definition of context free grammar briefly here.

\smallskip

\noindent
A {\it Context Free Grammar} is defined by a tuple $G = (N, T, R, S)$ where $N$ and $T$ are finite set of non-terminals and terminals respectively, $S$ is an initial non-terminal.
$R$ is a finite set of production rules of the form $(i) A \to BC$, $(ii) A \to a$, or $(iii) S \to \epsilon$ where $A \in N$, $B, C \in N \setminus \{S\}$ and $a \in T$.
We define derivation relation $\Rightarrow$ as following: Given $u_1, u_2 \in (N \cup T)^*$, $u_1 \Rightarrow u_2$ iff there exists $A \to w$ in the production rules set $R$ such that $u_1 = vAv'$ and $u_2 = vwv'$ for some $v,v' \in (N \cup T)^*$.
Let $\Rightarrow^+$ denote one or more application of derivation relation $\Rightarrow$.
A string $w$ belongs to the language generated by grammar $G$ iff it can be generated starting from $S$, applying derivation rules one or more times.
Hence, the language generated by the grammar $G$ is $\Lang{G} = \{w \in T^* ~|~ S \Rightarrow^+ w\}$.

\smallskip
\noindent
{\bf Proof of Theorem~\ref{sep:pda}}
We show that $\pptl$ separability problem of PDA is \pspacec{}. 
Given PDA $\Aa_1$ and $\Aa_2$, consider their corresponding CFGs $G_1, G_2$.
$\Lang{G_1} \upward \cap \Lang{G_2} \neq \emptyset$ iff $\Lang{G_1} \upward \cap \Lang{G_2} \downward \neq \emptyset$. 
If there is a witness of $\Lang{G_1} \upward \cap \Lang{G_2}$, the same witness trivially works for $\Lang{G_1} \upward \cap \Lang{G_2} \downward$. 
% as $L \subseteq L \downward$ for any language $L$.
Now suppose $w \in \Lang{G_1} \upward \cap \Lang{G_2} \downward$, then there exists $w' \in \Lang{G_2}$ s.t. $w \subword w'$. Also, $w' \in \Lang{G_1} \upward$. Hence, we have a witness of the non-emptiness of $\Lang{G_1} \upward \cap \Lang{G_2}$.	
We  prove that the non-emptiness of $\Lang{G_1} \upward \cap \Lang{G_2} \downward$ is $\mathsf{PSPACE-complete}$.

%\begin{itemize}
%\item[]
\noindent{$\mathsf{PSPACE}$-membership}.  The non-emptiness check of $\Lang{G_1} \upward \cap \Lang{G_2} \downward$ can be seen to be in $\mathsf{PSPACE}$ as follows. 	We construct  PDAs $A_{G_1 \upward}$, $A_{G_2 \downward}$ respectively for $\Lang{G_1} \upward$ and $\Lang{G_2} \downward$. 
This construction takes polynomial time;  moreover, 	$A_{G_1 \upward}$, $A_{G_2 \downward}$ 
use only a bounded stack which is polynomial in the size of the CFG ($G_1$ or $G_2$). 
This fact is rather easy to see for $A_{G_1 \upward}$ (Lemma~\ref{lem:upward}), but is 
quite involved for $A_{G_2 \downward}$ (Lemma \ref{lem:downward}). Assuming the construction 
of $A_{G_1 \upward}$, $A_{G_2 \downward}$, 	we give an $\mathsf{NPSPACE}$ algorithm as follows. 
Guess a word $w$, one symbol at a time, and run $A_{G_1 \upward}$, $A_{G_2 \downward}$  in parallel. 
Since both the PDAs require only a polynomially bounded stack size, we need polynomial space to store 
the information pertaining to the states, stacks and the current input symbol.  
If 
$A_{G_1 \upward}$, $A_{G_2 \downward}$ accept $w$, we are done. 
The $\mathsf{PSPACE}$-membership follows from Savitch's theorem. 

%\item[]
\noindent{$\mathsf{PSPACE}$-hardness}.  We reduce the halting problem of  a Linear Bounded Turing Machine (LBTM) to our problem. Given a 
LBTM $M$ and a word $w$, we  construct PDAs $P_1$ and $P_2$ such that $\Lang{P_1} \upward \cap \Lang{P_2} \downward\neq \emptyset$  iff $M$ accepts $w$. 
Our construction is in Appendix \ref{app:pspace-hard}. 

\subsection*{Bounded Stack Size of $A_{G \upward}$, $A_{G \downward}$}
%Let $G_i$ be the CFGs equivalent to PDA $\Aa_i$ such that $\langof{G_i} = \langof{\Aa_i}$ for $i = 1, 2$ so the $A_{\Aa_1 \upward}$ and $A_{\Aa_2 \downward}$ are same as $A_{G_1 \upward}$ and $A_{G_2 \downward}$ respectively.
%
%We first consider the case of $A_{\Aa \upward}$. 
%Let the equivalent CFG be $G$ s.t. $\langof{G} = \langof{\Aa}$.
We first consider the case of $A_{G \upward}$. The proof follows by examining the height of derivation trees of words in $\mini(\Lang{G})$, and 
showing that they have a bounded height. This bound gives the height of the stack size 
in $A_{G \upward}$. The proof of Lemma \ref{lem:upward} is in Section~\ref{app:up}. 

\begin{lemma}\label{lem:upward}
	Given a CFG $G = (N, T, R, S)$, the PDA $A_{G \upward}$ can be constructed in polynomial time, using a stack whose height is 
	polynomial in the size of $G$. $\Lang{A_{G \upward}} = \Lang{G} \upward$.
\end{lemma}

\begin{lemma}\label{lem:downward}
	Given CFG $G = (N, T, R, S)$, one can construct PDA $A_{G \downward}$ in polynomial time which uses polynomially bounded stack for computation such that $\Lang{A_{G \downward}} = \Lang{G} \downward$.
\end{lemma}
\begin{proof}
	First, we compute a CFG $G'$ using $G$, and then construct $A_{G \downward}$ using $G, G'$. The construction 
	of $G'$ is described below. 	We recall the approach described in~\cite{DBLP:journals/eatcs/Courcelle91} to compute downward closure, where 
	a regular expression is computed for $\Lang{G} \downward$. 
	
	For every language $L$, let $\alpha(L)$ be the set of terminal symbols occurring in $L$ (hence $\alpha(L) = \emptyset$ iff $L \subseteq \{\epsilon\}$).
	For $L, L' \subseteq T^*$ we have:
	(i) $\alpha(L \cup L') = \alpha(LL') = \alpha(L) \cup \alpha(L')$, (ii) $(L \cup L') \downward = L \downward \cup L' \downward$, and 
	(iii) $(LL') \downward = L\downward ~L'\downward$. 
	For $m \in (N \cup T)^*,$  let $L(G, m) $ denote the language generated by $G$, starting from the word $m$. For every $A, B \in N$,  define:	
	\begin{enumerate}
		\item $B <_1 A$ iff $A \Rightarrow^+_{G} mBm'$ for some $m,m' \in (N \cup T)^*$,
		\item $B <_2 A$ iff $A \Rightarrow^+_{G} mBm'Bm''$ for some $m,m',m'' \in (N \cup T)^*$,
		\item $B =_1 A $ iff $B <_1 A <_1 B$.
	\end{enumerate}
	A proof of the first implication in Claim \ref{claim:A2} is in Appendix \ref{app:claimm}, the second implication has a similar proof. 	
	\begin{claim}\label{claim:A2}
		$A <_2 A \Rightarrow L(G, A) \downward = (\alpha(L(G, A)))^*$, 
		$A =_1 B \Rightarrow L(G, A) \downward = L(G, B) \downward$
	\end{claim}
	%		 Suppose $w \in L(G, A) \downward$, then $w$ is labelled by symbols from $\alpha(L(G, A))$. Clearly $w \in (\alpha(L(G, A)))^*$.
	%		
	%		Conversely, assume $w \in (\alpha(L(G, A)))^*$. We prove by induction on the length of $w$ that $w \in L(G, A) \downward$. If the length of $w$ is $1$, then trivially one can generate any word in $L(G, A)$, which contains $w$ as a subword. Assume the length of $w$ to be $k+1$. By induction hypothesis, there exists a word $w' \in L(G, A)$ such that $w[:k] \subword w'$ and $w'$ is generated from $A$. We construct $w'' \in L(G, A)$ such that $w \subword w''$. The derivation of $w''$ from $A$ proceeds as follows.  Start with  the rule $A \Rightarrow^{+} mAm'Am''$, 
	%		and substitute the first $A$ with $w'$. For  the second $A$, substitute any word $w'''$ s.t. $A \Rightarrow^+_{G}w'''$ 
	%				s.t. the $k+1$st symbol of $w$,  
	%		 $w[k+1]$ is contained in  $w'''$. Then we obtain 	
	%		 $w \preceq w''$	, and hence $w \in 	L(G, A) \downward$.	 
	%
	%\noindent	In a similar manner, we can prove Claim \ref{claim:AeqB}. 
	%	\begin{claim}\label{claim:AeqB}
	%		If $A =_1 B$, then $L(G, A) \downward = L(G, B) \downward$.
	%	\end{claim}
	\noindent{\bf {Computing $L(G, A) \downward$}}. Now, we explain how $L(G, A) \downward$ can be computed for any given $A \in N$.
	If $A <_2 A$, then claim~\ref{claim:A2} yields the answer.
	Otherwise, we compute $L(G, A) \downward$ in terms of $L(G, B) \downward$ for $B <_1 A$ and $B \neq_1 A$, assuming 
	$L(G, B) \downward$ is given by previously computed rules for $L(G, B)$. 	Let $p: A \to m$ be a production rule.  We define the words $R_0(p), R_1(p)$ and $R_2(p)$ as follows, depending on the production rule  $p: A \to m$. 
	
	\noindent \textbf{First case:} $m$ does not contain any non-terminal $B =_1 A$. We let $R_0(p) := m$, and $R_1(p)$ and $R_2(p)$ be the empty words.  % In this case, we compute $L(G, A) \downward$ in terms of $L(G, B) \downward$ for $B <_1 A$ and $B \neq_1 A$. 
	
	\noindent \textbf{Second case:} $m$ contains a unique non-terminal $B$ with $B =_1 A$ and $m = m'Bm''$. We let $R_1(p) := m'$ and $R_2(p) := m''$. Since we assume that $A \nless_2 A$, the word $m$ cannot contain two occurrences of non-terminals $=_1 A$. 
	In this case,  $R_0(p)$ is the empty word.
	
	\begin{claim}\label{claim:Adownward}
		For every $A$ such that $A \nless_2 A$, we have:
		$$L(G, A) \downward = (\cup \;\alpha(L(G, R_1(p))))^* ~(\cup L(G, R_0(p))\downward ) ~(\cup\; \alpha(L(G, R_2(p))))^*$$ 	where the union extend to all production rules $p$ with lefthand side $B$ such that $B =_1 A$.
	\end{claim}
	
	Since the words $R_0(p), R_1(p)$ and $R_2(p)$ contain only non-terminals $C$ with $C <_1 A$ and $C \neq_1 A$, we have achieved our goal. By this we end the brief recall of the the approach described in~\cite{DBLP:journals/eatcs/Courcelle91} to compute downward closure.
	
	\noindent Based on the above facts, we can construct intermediate CFG $G' = (N', T', R', S\downward)$, which helps to construct PDA for $\Lang{G} \downward$. 	For each non-terminal $A \in N$, we introduce five non-terminals $A \downward, A_l, A_m, A_r$ and $A_{\alp}$ as well as dummy terminals $a_l$, $a_r$ and $a_{\alp}$ in $G'$.  Construction of $G'$ is shown in Algorithm~\ref{algo:downward} in Appendix \ref{app:down}. 
	The non-terminal $A_{\alp}$ as well as the terminal $a_{\alp}$ are 
	used to simulate $\alpha(L(G, A))$. Likewise, the 
	non-terminal $A_l$ ($A_r$) and the terminal  $a_l$ ($a_r$) are used to simulate $\alpha(L(G,m'))$ ($\alpha(L(G,m''))$) when $m=m'Cm''$ in the production 	 $A \to m$ and $C =_1 A$.  The non-terminal $A_m$ is used when $m=A_1A_2$ for non-terminals $A_1 \neq_1 A, A_2 \neq_1 A$. 
	Finally, 
	$A\downward$ works in place of $A$, and depending on $m$ 
	in the production $A \to m$, gets rewritten either as 
	$A_{\alp}$ or $A_l A_m A_r$ or $A_1 \downward A_2 \downward$. 
	
	\noindent{\bf{Constructing the PDA $A_{G\downward}$ from $G, G'$}}. 
	The PDA $A_{G\downward} {=} (Q, T, N', \delta, \{q\}, S \downward, \{q\})$ from $G'$ and $G$ is constructed below. 
	This works in the standard way of converting CFGs to PDA, keeping in mind the following. 
	When $a_l$ (resp. $a_r$, $a_{\alp}$) is the top of the stack, 
	the PDA has a loop over symbols from $\alpha_l(A)$ (resp. $\alpha_r(A)$, $\alpha(L(G, A))$), in whichever state it is, at that time. 
	$\alpha_l(A)$ (resp. $\alpha_r(A)$) is a set of symbols appearing in $\alpha(L(G,m'))$ (resp $\alpha(L(G,m''))$) for all rules $C \to m'Bm''$ where $C=_1 A=_1B$. 
	The PDA stays in state $q$ when the top of stack contains a non-terminal $A$. This non-terminal is popped and the  right hand side $m$ of its production rule $A \to m$ is pushed on the stack. If the top of the stack is i) a terminal from $T$ which also happens to be the next input symbol, then we just pop it ii) a terminal of the form $a_{\alp}$, $a_l$ or $a_r$, we move to state $q_{A_\alp}$, $q_{A_l}$ or $q_{A_r}$ respectively and loop over $\alpha(L(G, A))$, $\alpha_l(A)$ or $\alpha_r(A)$ respectively and come back to $q$ non-deterministically.
	The PDA accepts in $q$ when the stack becomes empty. 
	$A_l $, $A_r$ and $A_\alp$ do not produce any further non-terminals while $A_m$ produces non-terminals $B$ where  $B <_1 A$ and $B \neq_1 A$. Thus, the stack size is $\leq \mathcal{O}(|N|)$.
\end{proof}

\subsection{$\mathsf{PSPACE}$-hardness in Theorem \ref{sep:pda}}
\label{app:pspace-hard}
	Our construction is inspired from ~\cite[Lemma 8.6]{DBLP:books/aw/HopcroftU79}. 
	
	Let $M = (Q, \Sigma, \Gamma, \delta, q_0, B, F)$ be a linear bounded TM (LBTM). $\Sigma$ is the input alphabet, $\Gamma$ is the tape alphabet (contains $\Sigma$ and $B$ and some extra symbols)  and $B$ stands for the blank symbol.  
	$M$ uses at most $p(n)$ space for the input of size $n$, where $p(n)$ is a polynomial in $n$. 
	The maximum number of distinct configurations of $M$ can be ${max(|Q|, |\Gamma|)}^{p(|w|)+1}$, say $k'_{\max}$, 
	assuming a configuration is represented as $w_1 q w_2$ where $w_1w_2$ is a word over $\Gamma$ 
	that covers the entire tape (thus, we also include the spaces occupied by the blank symbols).   
	
	 Consider constants $k_{\max} = 2^{\lceil{\log k'_{\max}}\rceil}$ and $c_{\max} = \log{k_{\max}}$ to give precise bounds.
	Without loss of generality we assume that all halting computations of $M$ on $w$ take an even number of steps. 
	
	Given the LBTM $M$, we give the construction of PDAs $P_1$ and $P_2$ as follows. The input alphabet of $P_1, P_2$ 
	is $\Gamma \cup Q \cup \{\#,\$\}$, where $\#, \$$ are new symbols. The stack alphabet 
	for $P_1, P_2$ contain these symbols as well as symbols $0,1, \#_1$. 0,1 are used to encode binary numbers 
	to keep track of the number of configurations seen so far on any input. The idea is to restrict this number 
upto $k_{\max}$. 
	
		$P_1$ and $P_2$ both generate words of the form $c_1 \# c_2^R \# c_3 \# c_4^R \# \ldots \#c^R_k \# \$^{c_{max}} \ldots \# \$^{c_{max}} \#$ where number of $\#$s = $k_{\max}$, satisfying the following conditions. 
			\begin{itemize}
		\item[(1)] each $c_i$ is a valid configuration of $M$ and of size exactly ${c_{\max}}$,
		\item[(2)] In $P_2$, $c_1$ is an initial configuration of $M$ on $w$ i.e. $q_0wB^{c_{\max} - (|w|+1)}$,
		\item[(3)] In $P_1$, the last configuration before $\$$ is an accepting configuration, 
		\item[(4)]  $c_{i+1}$ is a valid successor configuration of $c_i$ in $P_1$ for all odd $i$,  
		\item[(5)] $c_{i+1}$ is a valid successor configuration of $c_i$ in $P_2$ for all even $i$. 
	\end{itemize} 
	This way, correctness of consecutive odd and even configurations are guaranteed by $P_2$ and $P_1$ respectively and $\Lang{P_1} \cap \Lang{P_2}$ is the set of accepting computations of $M$ on $w$ padded with copies of word ($\$^{c_{\max}} \#$) at the right end. 
The length of the words accepted by $P_1$ and $P_2$ is $k_{max}$ (number of $\#$s) + $c_{max}.k_{max}$ (the size of each $c_i$ or $c_i^R$ 
is $c_{max}$, the padding with $\$$ also has length $c_{max}$. There are $k_{max}$ of these $c_i$ or padded $\$$.)

\noindent{\bf{Construction of $P_1$}}. 		
	\begin{enumerate}
	\item [(a)]	
	Initially we push $\log k_{\max}$ zeros (where $k_{\max}$ is the maximum possible number of configurations) and a unique symbol $\#_1$ in the stack to keep track of number of $\#$s read so far. $\#_1$ is used as a separator in the stack between 
	the counter values (encoded using 0,1) and the encoded configurations over $\Gamma, Q$.   Thus, we initialize 
	the value of the binary counter to $\log k_{\max}$ zeroes, with the separator $\#_1$ on top. 
	\item [(b)] Let $L$ be $\{c_1 \# c_2^R \# ~|~ c_2$ is the successor configuration of $c_1 \text{~in~} M \}$. 
	To realize $L$, we do the following. 
We read $c_1$ and check whether it is of the form $w_1 q w_1'$ where $w_1, w_1' \in \Gamma^*$ and $q \in Q$. 
	While reading $c_1$ upto $q$, we push $w_1$ on its stack. As soon as we find $q \in Q$ in $c_1$, we store $q$ in the finite control state and read the next input symbol, say $X$. 
	%(if the next symbol is $\#$, take $X$ to be $B$).
	\begin{itemize}
	\item 	
	If $X$ is $\#$, then there is no successor configuration, and we reject the word denoting that we have reached the end of the tape. 
	\item If $\delta(q, X) = (p, Y, R)$, then we push $Yp$ onto the stack. 
	\item If $\delta(q, X) = (p, Y, L)$, let $Z$ be on top of stack, then we replace $Z$ by $pZY$.
	\end{itemize}
	% (but if the input last read was $\#$, and $Y = B$, just replace $Z$ by $pZ$, or by $p$ if $Z$ is also $B$).
	Again while reading $w_1'$, we push the symbols read onto the stack. 
	%	After reading $w_1'$, $P_4$ pushes extra blank symbols $B$ 
	Now the stack content will be (from bottom to top) the encoding of the counter, $\#_1$, and $c_2$. 
	After reading the $\#$ after $c_1$, we compare each input symbol with the top stack symbol (now we are  reading 
	$c_2^R$ in the input). 	 If they differ, there is no next move.  If they are equal, we pop the top stack symbol. When the top of stack is $\#_1$, 
	we read $\#$ after $c_2^R$ and accept.
	While reading $c_i$ (or $c_i^R$) we must also verify that the length of $c_i$ is exactly $c_{\max}$.  
	We can do this by adding polynomial size counter in the finite control state. 
	\item[(c)] We also need to increment the counter by 2 for the number of $\#$ (configurations) seen so far. 
	We have read two configurations now in the above step.  This can be done by popping the stack content starting from $\#_1$ until we find $0$. During this popping, we store in the finite control, the information of how many 1s we have popped. As soon as we find $0$ on the  top of the stack, replace it with $1$ and push back as many number of $1$s as we popped previously (this info is stored in the finite control), 
	 followed by the separator $\#_1$. We need to do this step again to increment the counter since we saw two configurations (the two $\#$s represent this) in $L$. 
	  
	\end{enumerate}
To construct the PDA $P_1$, we follow 
	steps of [(a)] i.e., initializing the counter. Next it iterates over steps of [(b)] followed by [(c)], thus reading two consecutive $c_i \# c_{i+1}^R \#$ and incrementing the counter by $2$ until we find the configuration of the form $\Gamma^* F \Gamma^*$ or the counter becomes full. 
	If we stop the iteration due to accepting configuration and counter is not full yet  (counter is full when we have all $1$s in the stack),  we pad $\$^{c_{max}} \#$ until it becomes full.  
	Otherwise if we stop the iteration due to the counter reaching $k_{\max}$ value before getting accepting configuration, we stop and reject this sequence of configurations.

\noindent{\bf{Construction of $P_2$}}. 		
		Similarly, we can construct PDA $P_2$: The difference here is that it starts with initial configuration i.e. 
		$q_0 wB^{c_{\max} - |w|} $ followed by $\#$ (here also we increment the counter by $1$), then it will iterate over $\{c_1^R \# c_2 \# ~|~ c_2 \text{~is ~the~ successor~ configuration~ of~ } c_1 \text{~in~} M \}$ and increment the counter by $2$ after each iteration. Then it will read accepting configuration non-deterministically (of course, followed by $\#$ and increment in the binary counter) and pad the string with $\$^{c_{max}}\#$ until counter becomes full.
	
	By  padding the words accepted with extra $\$^{c_{max}}\#$s we are enforcing that $P_1$ and $P_2$ accepts words of same length. 
	Having the  blank symbols $B$ as part of the LBTM configuration is also motivated by this. 
	Indeed, $w \in \Lang{P_1} \cap \Lang{P_2}$ iff the LBTM $M$ accepts $w$. 
	Since the words accepted by $P_1, P_2$ have same constant length, 
	$\Lang{P_1} \cap \Lang{P_2} \neq \emptyset$ 	iff $\Lang{P_1} \upward \cap \Lang{P_2} \downward \neq \emptyset$. 

\subsection{Proof of Claim \ref{claim:A2}}
\label{app:claimm}
		 Suppose $w \in L(G, A) \downward$, then $w$ is labelled by symbols from $\alpha(L(G, A))$. Clearly $w \in (\alpha(L(G, A)))^*$.
		
		Conversely, assume $w \in (\alpha(L(G, A)))^*$. We prove by induction on the length of $w$ that $w \in L(G, A) \downward$. If the length of $w$ is $1$, then trivially one can generate any word in $L(G, A)$, which contains $w$ as a subword. Assume the length of $w$ to be $k+1$. By induction hypothesis, there exists a word $w' \in L(G, A)$ such that $w[:k] \subword w'$ and $w'$ is generated from $A$. We construct $w'' \in L(G, A)$ such that $w \subword w''$. The derivation of $w''$ from $A$ proceeds as follows.  Start with  the rule $A \Rightarrow^{+} mAm'Am''$, 
		and substitute the first $A$ with $w'$. For  the second $A$, substitute any word $w'''$ s.t. $A \Rightarrow^+_{G}w'''$ 
				s.t. the $k+1$st symbol of $w$,  
		 $w[k+1]$ is contained in  $w'''$. Then we obtain 	
		 $w \preceq w''$	, and hence $w \in 	L(G, A) \downward$.	 

\subsection{Construction of $G'$ and $A_{G\downward}$}	
\label{app:down}

%\begin{wrapfigure}[25]{l}
	\begin{algorithm}[t]
	%	\SetAlgoLined
		%			\KwIn{$G_2$}
		%		\KwData{this text}
		%		\KwResult{how to write algorithm with \LaTeX2e }
		%		initialization\;
		\ForEach{\text{non-terminal} $A \in N$}{
			$\alpha_l(A):= \emptyset, \alpha_r(A):= \emptyset$;
		}
		\ForEach{\text{non-terminal} $A \in N$}{
			\eIf{$A <_2 A$}{
				$A\downward \to  A_{\alp}$ \\
			}
			{
				\ForEach{$B \in N$ s.t. $B =_1 A$ }{
					\ForEach{$p: B \to m$}{
						\If{$m = A_1 A_2$ for $A_1 \neq_1 A$ and $A_2 \neq_1 A$}{
							$A_m \to A_1 \downward A_2 \downward $
						}
						\If{$m = a$ for $a \in T$}{
							$A_m \to a ~|~ \epsilon$
						}
						\If{$m = m'Cm''$ for unique $C =_1 A$}{
							$\alpha_l(A):= \alpha(L(G_2, m')) \cup \alpha_l(A),~~~$
							$\alpha_r(A):= \alpha(L(G_2,m'')) \cup \alpha_r(A)$ 						
						}
					}
				}
				$A \downward \to A_l ~ A_m~ A_r,~~$ 
				$A_l \to a_l,~~$ 
				$A_r \to a_r,~~$ 
				$A_{\alp} \to a_\alp. $
			}
		}
		\caption{Construction of $G'$}
		\label{algo:downward}
	\end{algorithm}
%	\end{wrapfigure}

\noindent{\bf{The PDA $A_{G\downward}$ from $G, G'$}}. 
Now we construct the PDA $A_{G\downward} = (Q, T, N', \delta, \{q\}, S \downward, \{q\})$ from $G'$ and $G$. This works in the standard way of converting CFGs to PDA, keeping in mind the following. 
When $a_l$ (resp. $a_r$, $a_{\alp}$) is the top of the stack, 
the PDA has a loop over symbols from $\alpha_l(A)$ (resp. $\alpha_r(A)$, $\alpha(L(G, A))$), in whichever state it is, at that time. 
The PDA stays in state $q$ when the top of stack contains a non-terminal $A$. This non-terminal is popped the  right hand side $m$ of its production rule $A \to m$ is pushed on the stack. If the top of the stack is i) a terminal from $T$ which also happens to be the next input symbol, then we just pop it ii) a terminal of the form $a_{\alp}$, $a_l$ or $a_r$, we move to state $q_{A_\alp}$, $q_{A_l}$ or $q_{A_r}$ respectively and loop over $\alpha(L(G, A))$, $\alpha_l(A)$ or $\alpha_r(A)$ respectively and come back to $q$ non-deterministically.
	%	Formally, $\delta':$
	%	\begin{itemize}
	%		\item $(q, \epsilon, A, q, m) \in \delta'$ if $A \to m$ is production rule in $G_2'$ 
	%%		\item $(q, a, A, q, \epsilon) \in \delta'$ for all terminals $a \in T \cup \{\epsilon\}$ of $G_2$ if $A \to a$ is production rule in $G_2'$ (when the top of the stack contains a non-terminal whose rule has terminals from $G_2$)
	%		\item $(q, \epsilon, A_{\alp}, q_{A_{\alp}}, \epsilon) \in \delta'$
	%		\item $(q, \epsilon, A_l, q_{A_l}, \epsilon) \in \delta'$ 
	%		\item $(q, \epsilon, A_r, q_{A_r}, \epsilon) \in \delta'$
	%		\item $(q_{A_\alp}, \epsilon, Z, q, \epsilon) \in \delta'$ for any symbol $Z$ as top of the stack
	%		\item $(q_{A_l}, \epsilon, Z, q, \epsilon) \in \delta'$ for any symbol $Z$ as top of the stack
	%		\item $(q_{A_r}, \epsilon, Z, q, \epsilon) \in \delta'$ for any symbol $Z$ as top of the stack
	%		\item $(q_{A_\alp}, a, Z, q_{A_l}, \epsilon) \in \delta'$ for any symbol $Z$ as top of the stack and $a \in \alpha(L(G_2, A))$
	%		\item $(q_{A_l}, a, Z, q_{A_l}, \epsilon) \in \delta'$ for any symbol $Z$ as top of the stack and $a \in \alpha_l(A)$
	%		\item $(q_{A_r}, a, Z, q_{A_r}, \epsilon) \in \delta'$ for any symbol $Z$ as top of the stack and $a \in \alpha_r(A)$
	%	\end{itemize}
	The PDA accepts in $q$ when the stack becomes empty. 
	
\noindent{\bf{Complexity}}. We analyze the complexity of constructing $G'$ as well as $A_{G \downward}$. 
	\begin{itemize}
		\item Construction of $G'$:
		\begin{itemize}
			\item for every $A, B \in N$, checking $A <_2 B$ reduces to checking membership of $a_1 a_2$ in $L(G_1, B)$ where $G_1$ 
			is obtained from $G$ by replacing productions $C \to a$ of $G$ with $C \to \epsilon$, $C \to a_1 ~|~ a_2$. $G_1$ is constructed in polynomial time from $G$.  Similarly, we can check if $A <_1 B$ and $a \in \alpha(L(G, A))$. 
			\item for every $A \in N$, five non-terminals and three terminals are introduced. Hence, $G'$ has $5|N|$ non-terminals and $|T| + 3|N|$ terminals, which is again polynomial in the input $G'$. 
			\item each step of the algorithm~\ref{algo:downward} can be computed in at most polynomial time and the algorithm has loops whose sizes range over terminals, non-terminals and production rules of $G$.
		\end{itemize}
		\item PDA $A_{G \downward}$: Each  non-terminal $A$ has corresponding non-terminals  $A_l$, $A_r$, $A_\alp$ or $A_m$ in $G'$. Out of these,   
		 $A_l $, $A_r$ and $A_\alp$ do not produce any further non-terminals while $A_m$ produces non-terminals $B$ where  $B <_1 A$ and $B \neq_1 A$. Thus, the stack size is bounded above by $\mathcal{O}(|N|)$.
	\end{itemize}

	\subsection{Proof of Lemma \ref{lem:upward}}
	\label{app:up}
	Let $G = (N, T, R, S)$, we construct CFG $G'$ which accepts subset of $\Lang{G} $ and includes  $\mini(\Lang{G})$. Clearly $\Lang{G} \upward = \Lang{G'}\upward$.
	Let $G'$ be $(N \times \{0, \ldots, |N|\}, T, R', (S, |N|))$ where production rules are defined as following:
	\begin{itemize}
		\item $(A, 0) \to a $ if $A \to a$ 
		\item $(A, i) \to (B, i-1) (C,i-1) $ for all $i \in [1, |N|]$ if $A \to BC $ 
		\item $(A, i) \to (A, i-1)$ for all $i \in [1, |N|]$
	\end{itemize}
	We can construct a PDA $P$ for $\Lang{G'}$ in polynomial time. It is easy to see that the PDA uses at most $|N| + 1$ height of the stack. It accepts all words of $\Lang{G}$ having derivation with at most $|N|+1$ steps. We claim that
	$\Lang{G'}$ contains $\mini(\Lang{G})$. 
	If not, then there exists $w \in \mini(\Lang{G})$ with at least $|N|+2$ derivation steps i.e.  some non-terminal $A$ appears at least two times in derivation tree say at level $i$ and $j > i$. If we replace the subtree rooted at level $j$ by subtree at level $i$, we get a proper subword of $w$ (wlog we assume $G$ is in CNF and it does not have useless symbols and $\epsilon-$productions), which is also in $\Lang{G}$, contradicting our assumption that $w \in \mini(\Lang{G})$. 
	Finally we can obtain a PDA for $\Lang{G} \upward $ by adding self loops on each state of $P$ for each symbol $a \in \Sigma$ as $\Lang{G} \upward = \Lang{G'} \upward$.

\subsection{Proof of Lemma~\ref{lem:sl2ompl-pptl} }

	Given two SL formulas $\formula$ and $\formula'$, with $x_1, \ldots,x_n$ as their set of variables. Let   $\Aa$ and $\Aa'$ be two OMPAs such that $\Lang{\Aa}=\encode{(\langof{\formula})}$ and  $\Lang{\Aa'}=\encode{(\langof{\formula'})}$. Then, $\formula$ and $\formula'$ are separable by an $n$-$\pptl$ iff  $\Aa$ and $\Aa'$ are separable by a $\pptl$.

\begin{proof}
	Let $\langof{\Aa}$ and $\langof{\Aa'}$ be separable by a $\pptl$ $L$ over $\Sigma \cup \{\#\}$.
	Consider a regular language $R \subseteq (\Sigma \cup \{\#\})^*$ containing all words having exactly $(n-1)$ $\#$s. 
	We first show that $L \cap R$  can be written as a finite Boolean combination, except negation, of  languages of the form $L_1 \# L_2 \# \ldots \# L_n$ where each $L_i$ is a piece language over $\Sigma$. 
	
 We prove inductively the representation of $L \cap R$ in terms of positive Boolean combinations of 
      $n$ piece languages over $\Sigma$ having $\#$ in between them. 
      As a base case, let $L$ be a piece language $\Sigma^* a_1 \Sigma^* \ldots \Sigma^* a_k \Sigma^*$. 
	Let $S$ be a finite set containing only the words of the form $w$ such that $a_1 a_2 \ldots a_k \preceq w$  and the symbol $\#$ appears exactly  $(n-1)$-times in  $w$. 
	Then $L \cap R$ can be written as $\bigcup\limits_{w \in S} \Sigma^* b_1\Sigma^* b_2 \ldots b_\ell \Sigma^*$, where $w{=}b_1 \dots b_{\ell}$, which is of the form 
	$\bigcup\limits_{i=1}^{|S|} L_{i1} \# L_{i2} \# \ldots \# L_{in}$ where $L_{ij}$ is a piece language over $\Sigma$.
	
Now assume that $L$ is of the form $L_1 \cap L_2$ (resp. $L_1 \cup L_2$). It is easy to see that $L \cap R$ is equivalent to 	 $(L_1 \cap R) \cap   (L_2\cap R)$ (resp. $(L_1\cap R) \cup (L_2 \cap R)$). Thus we can use our  induction hypothesis to show that $L \cap R$ is indeed a Boolean combination (using union and intersection) of languages of the form  $L_1 \# L_2 \# \ldots \# L_n$ where $L_i$ is a piece language over $\Sigma$. 

Now, we are ready to construct an $n$-$\pptl$ that separates  $\langof{\formula}$ and $\langof{\formula'}$. First of all, observe that $\Lang{\Aa}\subseteq L \cap R$ and $\Lang{\Aa'}\cap  (L \cap R) = \emptyset$ (from the definition of $\Aa$ and $\Aa'$). Since $L \cap R$ can be written as a positive Boolean combination of   languages of the form $L_1 \# L_2 \# \ldots \# L_n$ where $L_i$'s are piece languages over $\Sigma$, we can easily construct an $n$-$\pptl$ over $\Sigma$  that separates  the two formulas  $\formula$ and $\formula'$ as a finite positive Boolean combination of the product of these $L_i$s. This is similar to Lemma \ref{SL2OMPA}, where we used the fact that 
a finite Boolean combination of the product of piece languages is a $\ptl$ over an $n$-tuple alphabet.  
 The proof of the other direction follows similarly.\end{proof}

\subsection{Proof of claim in section \ref{subsec:rssl}, line 421}
\paragraph*{Converting SL with functional transducers to right-sided SL}

\begin{lemma}
Let $\formula$ be an SL formula where all relational constraints 
have only functional transducers. Then we can obtain 
a right-sided SL formula $\formula'$ such that $\formula'$ is satisfiable 
iff $\formula$ is.	
\end{lemma}

%
%\begin{lemma}\label{lem:sl-normal-form}
%	Let $\formula$ be SL formula of the form $\bigwedge\limits_{i=1}^{n} x_i \in \langof{\automaton_i} \wedge \bigwedge\limits_{i=1}^{k} \constr_i$ such that $\constr_i$ is of the form $(x_i,\term_i)  \in  \relof{\transducer_i}$ where each $\transducer_i$ defines a (partial) function. Moreover, if a variable $x_j$ is appearing in $\term_i$ then $j>i$. One can write  
%	equisatisfiable right sided formula $\psi' = \bigwedge\limits_{i=1}^{n} x_i \in \langof{\automaton_i} \wedge \bigwedge\limits_{i=1}^{k} \constr'_i$ such that if a variable $x_j$ is appearing in $\constr'_i$ then ${\bf j>k}$.
%\end{lemma}

\begin{proof}
Assume that the SL formula $\formula$ with functional transducers is given as $$\formula=\bigwedge\limits_{i=1}^{n} x_i \in \langof{\automaton_i} \wedge \bigwedge\limits_{i=1}^{k} \constr_i$$ such that $\constr_i$ is of the form $\Tt_i(\term_i)=x_i$, and  $\transducer_i$ defines a (partial) function. 
Without loss of generality, assume that each $t_i$ is a sequence of variables. By the SL condition,  if a variable $x_j$ appears in $\term_i$ then $j>i$.  
Notice that $\formula$  need not be rightsided, since we can have $\varphi_h$ as the constraint $x_h= \Tt_h(t_h)$ with $t_h=x_j$, $j > h$ and also 
$\varphi_j$ being $x_j=\Tt_j(t_j)$ violating the condition that none of the  output variables $x_1, \dots x_k$ can be present 
in any of the inputs $t_1, \dots, t_k$.  The variables $x_{k+1}, \dots, x_n$ are independent 
and do not appear in the left side (or outputs) of any of the transductions. %, in other terms, they are not assigned values using relational constraints.
	
	We first explain the idea of the proof. We construct a formula $\formula'$ in which the variables appearing in the inputs terms $t'_1, \dots, t'_k$ 
	of all the transductions  in $\formula'$ are independent variables of $\formula'$. 
	 		  Thus, each of $t'_{1}, \ldots, t'_{k}$ will  be sequences 
	of variables from $x_{k+1}, \ldots, x_n$. The idea is to apply substitution inductively starting from $\constr_{k}$ to $\constr_{1}$, obtaining 
	in $\formula'$, the constraints $\varphi'_k, \dots, \varphi'_1$ s.t. $\formula'$ are rightsided.
		Note that by the SL definition, $t_k$ can have 
	only variables from $x_{k+1}, \dots, x_n$. Now  consider $\varphi_{k-1}$ given by $x_{k-1}=\Tt_{k-1}(t_{k-1})$. We know that 
	$t_{k-1}$ can have only  variables from $x_k, x_{k+1}, \dots, x_n$. We can replace all occurrences 
	of $x_k$ in $t_{k-1}$ with $\Tt_k(t_k)$ obtaining $x_{k-1}=\Tt_{k-1}(t_{k-1}[ x_k \mapsto \Tt_k(t_k)])$. 
		Likewise, we can replace occurrences of $x_{k-1}, x_k$ 
	in $t_{k-2}$ respectively with $\Tt_{k-1}(t_{k-1})$ and $\Tt_k(t_k)$.  
We proceed iteratively for  $\constr_{k-3}$ to $\constr_1$.
	
	 We give the formal construction by induction on $n$. For the base case, 
	 $\varphi'_k$ is same as $\varphi_k$ since $t_k$ is over the independent variables 
	 $\{x_{k+1}, \dots, x_n\}$. The same holds for all $\varphi'_h$ s.t. the input $t_h$ 
	 is over $\{x_{k+1}, \dots, x_n\}$. 
	 
	 For the inductive step, let $j$ be the largest index s.t. the input $t_j$ in $\varphi_j$ contains dependent variables (from 
$\{x_{j+1}, \dots, x_k\}$). Let $\varphi_j$ be given by $x_j=\Tt_j(t_j)$ where $t_j=t_j[1] \dots t_j[m]$, where $t_j[1], \dots, t_j[m] \in \{x_{j+1}, \dots, x_k, x_{k+1}, \dots, x_n\}$. 
Consider an accepting run in $\Tt_j$ visiting states $q_0, \dots, q_m$ s.t. 
the part of the run between $q_{i-1}$ and $q_i$ processes the input $t_j[i]$.  If there are multiple accepting runs for the same input word in $\Tt_j$, 
	the sequence $s$ of states $q_0, \dots, q_m$ can differ, but the output will be the same since $\Tt_j$ is functional. 
	Let $\Tt_j(q_0, q_1), \dots, \Tt_j(q_{m-1}, q_m)$ represent 
	transducers induced from $\Tt_j$, s.t. $\Tt_j(q_0, q_1)$ reads the first variable $t_j[1]$ and produces 
	output (say $o_1$),  $\Tt_j(q_1, q_2)$ reads the second variable $t_j[2]$ and produces 
	output (say $o_2$), and so on till $\Tt_j(q_{m-1}, q_m)$ reads the $m$th variable $t_j[m]$ and produces 
	output (say $o_m$).  Let $\Tt_j^{s,i}$ denote the induced transducer $\Tt_j(q_{i-1}, q_i)$ 
	for $1 \leq i \leq m$ (i.e., the transducer $\Tt_j$ with $q_{i-1}$ as initial state and $q_i$ as a final one). Then, for a fixed sequence of states $q_0, \dots, q_m$,  we can 
	write $x_j$ as $\Tt_j^{s,1}(t_j[1])\Tt_j^{s,2}(t_j[2])\dots \Tt_j^{s,m}(t_j[m])$. Note that 
	the concatenation of  transducers  $\Tt_j^{s,1}, \dots,\Tt_j^{s,m}$ over inputs 
	$t_j[1], \dots, t_j[m]$ can be replaced by a single transducer $\Tt'_j$ over the input $t_j=t_j[1] \dots t_j[m]$ 
	producing output $o_1 \dots o_m$, since  transducers are closed under concatenation. 

 Assume that $t_j[\ell]=x_p$ is a dependent variable, $p \in \{j+1, \dots, k\}$.
  Then we can replace $\Tt_j^{s,\ell}(t[\ell])$ with 
$\Tt_j^{s,\ell}[\Tt_{p}(t_p)]$.
Now, we apply such substitution to all the variables appearing in $t$.
This will  result in that  $t'$ is a sequence  over independent variables $\{x_{k+1},\dots, x_n\}$.  

Thus, we can rewrite $x_j$ as  
$\Tt_j^{s,1}(\Tt'_{p_1}(t_{p_1}))\Tt_j^{s,2}(\Tt'_{p_2}(t_{p_2}))
\dots \Tt_j^{s,m}(\Tt'_{p_m}(t_{p_m}))$ where 
$p_1, \dots, p_m$ is defined such that $t_j[\ell]=x_{p_{\ell}}$ for all $\ell \in \{1,\ldots,m\}$ , and 
$\Tt'_{p_i}$ is the transducer given by $\Tt_{p_i}$ in case 
$p_i \in \{j+1, \dots, k\}$, and is the identity 
otherwise. Functional transducers are closed under composition, 
hence we can obtain a functional transducer $\Tt''_i$ equivalent to each of 
  the compositions $\Tt_j^{s,i} \circ \Tt'_{p_i}$, for $1 \leq i \leq m$. 
  This expression for $x_j$ is for a fixed accepting path 
through states $q_0, \dots, q_m$. Considering the general case of any accepting path over an accepting sequence $q_{i_0}\dots q_{i_m}$, we can write $\varphi''_j$ as 
$$\varphi''_j=\bigvee_{s=q_{i_0}, \dots, q_{i_m}}[
x_j=\Tt_j^{s,1}(\Tt'_{p_1}(t_{p_1}))\Tt_j^{s,2}(\Tt'_{p_2}(t_{p_2}))
\dots \Tt_j^{s,m}(\Tt'_{p_m}(t_{p_m}))]$$

Notice again that $t_{p_1}t_{p_2} \dots t_{p_m}$ is a sequence over independent 
variables, and the union and concatenation above can be replaced by a single transducer $\Tt_j^s$. $\Tt_j^s$ is the transducer obtained by the concatenation of the following transducers $(\Tt_j^{s,1}(\Tt'_{p_1})) \cdot \{(\epsilon, \$)\} \cdot (Tt_j^{s,2}(\Tt'_{p_2}))  \cdot \{(\epsilon, \$)\} \cdot 
\dots \cdot \{(\epsilon, \$)\} \cdot (\Tt_j^{s,m}(\Tt'_{p_m}))$ where $\$$ is a special symbol not in $\Sigma$. Now let $\Tt'_j$ be the transducer obtained by taking the union of $\Tt_j^s$ where $s$ is a sequence of states of $\Tt_j$ of length $m$. Now we can rewrite $\varphi''_j$ as $\varphi'''_j$ defined by $$
x_j=\Tt'_j(t_{p_1} \$t_{p_2} \$ \dots \$ t_{p_m})$$

We can now replace the occurrence of $\$$ by a fresh variable $z$ not appearing in $\formula$ to obtain the constraint $\varphi'_j$ defined by $$
x_j=\Tt'_j(t_{p_1} z t_{p_2} z \dots z t_{p_m})$$

% Further, $\Tt_j(t_j)$ in $\varphi_j$ is equivalent to the union of the concatenation above of transducers, since
%(1)  
%we have only decomposed any accepting run of $\Tt_j$ 
%on $t_j$ into corresponding subruns over $t_j[1], \dots, t_j[m]$ and stitched them together, (2) in case of encountering 
%a dependent variable $x_{\ell}$, $j< \ell \leq k$, we have replaced the part  
%$x_{\ell}$ of the input  with $\Tt'(t')$ where $x_{\ell}=\Tt'(t_\ell)$. 

%By the inductive hypothesis  at $\varphi'_{\ell}$, we have already shown the equivalence  of $\varphi_{\ell}$ and $\varphi'_{\ell}$. Thus, the part of $\Tt_j$  working on $x_j$ now works on the equivalent representation  $\Tt'(t')$, yielding the same output.

We now proceed with the next largest $j$ s.t. the input of $\varphi_j$ is replaced by  $\varphi'_j$ and has now
has a dependent variable. We repeat this procedure until we get rid of all the dependant variables as input of the relational constraints. Let $\formula''$ the resulting formula.
$$\formula''=\bigwedge\limits_{i=1}^{n} x_i \in \langof{\automaton_i} \wedge \bigwedge\limits_{i=1}^{k} \constr'_i$$
We can now define the formula $\formula'$ as the conjunction of $\formula''$  and $z \in \{\$\}$.

	Based on the construction of $\formula'$, we have the following claims.
	\begin{lemma}
	$\formula'$ is a rightsided SL formula.	
	\end{lemma}
 \begin{proof}
By construction, we have 
 $$\formula'= z \in \{\$\} \wedge \bigwedge\limits_{i=1}^{n} x_i \in \langof{\automaton_i} \wedge \bigwedge\limits_{i=1}^{k} \constr'_i$$
 
The variables $z,x_{k+1}, \dots, x_n$ are clearly independent in $\formula'$ as well. Consider a $\varphi'_j$. 
As described above, we know that the input of the transduction $\varphi'_j$ 
is a sequence over independent variables (thanks to composition of transducers).
%
%
%
%It is a conjunction 
%of functional transductions of the form  $x_j[p]=\Tt(t)$ where 
%$p$ varies from 1 to $\partition(x_j)$,  $t$ is a sequence of  independent variables from  
%$\{x_{k+1}, \dots, x_n\}$,  and $\Tt$ is a transducer obtained by composing $\leq k-j$ 
%transducers. Thus, all the left sides (or outputs) of transductions $\varphi_j'$ 
%in $\formula'$ are the fresh variables introduced, and 
%do not appear in the inputs of any of the transducers. 
All input variables are among $\{z,x_{k+1}, \dots, x_n\}$, the independent variables. 
Hence $\formula'$ is rightsided.
 \end{proof}

\begin{lemma}
	The formula $\formula'$ constructed is 
	satisfiable iff  $\formula$ is.	
\end{lemma}
	\begin{proof}

The equivalence of each $\varphi'_j$ with $\varphi_j$ ensures that 
we have the same functions computed by the transductions, and hence all 
the inputs and  outputs agree in each of the old and new transductions. 
This, along with the fact that independent variables remain untouched (with the exception of the fresh variable $z$ which has as evaluation $\$$), ensure that we  have the same valuation for $x_1, \dots, x_n$ in 
both $\formula$ and $\formula'$.
 
%
% 
%
%	
%	
% For each satisfying assignment 
%	$\eta(x_1), \dots, \eta(x_n)$ of $\formula$, we have the satisfying 
%	assignment for $\formula'$ as follows. Let $c_i$ = $\partition(x_i)$ for $1 \leq i \leq k$. Then 
%		$\eta(x_1[1]), \dots, \eta(x_1[c_1]), \dots, \eta(x_k[1]), \dots, \eta(x_k[c_k]), 
%		\eta(x_{k+1}), \dots, \eta(x_n)$ is a satisfying assignment 
%for $\formula'$ s.t. $\eta(x_i)=\eta(x_i[1]).\dots \eta(x_i[c_i])$ for $1 \leq i \leq k$. 
\end{proof}
Thus, we have shown that, given a SL string constraint having only functional transducers, 
we can compute an equisatisfiable rightsided SL string constraint.
\end{proof}

\subsection{Details for Lemma \ref{lem:sl-2nft}}
%Before reading this proof, readers are advised to read the proof of Lemma~\ref{lem:sl-2nft}.
%It is in continuation of Lemma~\ref{lem:sl-2nft}.
Given a right-sided SL formula $\formula$ over  $\alphabet$, with  $x_1,x_2,\ldots,x_n$ as its set of  variables, we wish to construct, in polynomial time, a $\tnft$ $\Aa_\formula$ such that $\sem{\Aa_\formula}{=}\{(u_1\#u_2 \# \cdots \#u_n, w_1 \# w_2 \# \ldots  \# w_n) {\mid}  u_1\#u_2 \# \cdots \#u_n \in \encode{(\langof{\formula})}$ and  $w_i{=}u_i$ if  $x_i$ is an independent variable $\}$.

To recall, $\formula$ is of the form $\bigwedge\limits_{i=1}^{n} y_i \in \langof{\automaton_i} \wedge \bigwedge\limits_{i=1}^{k}(y_i,\term_i)  \in  \relof{\transducer_i}$  with  $y_1,\ldots,y_n$ is a permutation of $x_1,\ldots,x_n$. 
 $\pi: [1,n] \rightarrow [1,n]$ is the mapping that associates to each index $i \in [1,n]$, the index $j \in [1,n]$ s.t. $x_i=y_j$ (or $x_i=y_{\pi(i)}$). 

As already mentioned above, given $n$ blocks separated by $\#$, $\Aa_\formula$ treats block $i$ as valuation of $x_i$ if $x_i$ is an independent variable, for all $i$.
For other (dependent) variables, say $x_j$, $\Aa_{\formula}$, computes its valuation based on the valuation of variables given in input blocks.
$\Aa_\formula$ considers the ordering of variables as $x_1, \ldots, x_n$ for input and output.
Thus if the input is $u_1 \# u_2 \# \ldots \# u_n$, then $\Aa_{\formula}$ tries to find satisfying assignment $\eta$ for $\formula$, and output for this input is given as $\eta(x_1) \# \ldots \# \eta(x_n)$ s.t. $\eta \models \formula$ and $\eta(x_i) = u_i$ for all independent variables $i$.

Working of $\Aa_\formula$: It produces valuation of $x_1, x_2, \ldots, x_n$ in sequence. 
Depending on if $x_i$ is dependent or independent variable, process given in Lemma~\ref{lem:sl-2nft} is followed by $\Aa_{\formula}$.
To execute the process, $\Aa_{\formula}$ uses following set of states and transitions.

{\noindent \bf Set of states:} 
\begin{itemize}
	\item For each $x_i$, if $x_i$ is an independent variable, we have following set of states used to produce $x_i$ in the output:
	\begin{itemize}
		\item $\{(q_{\vdash}(x_i), r) \mid r \text{~is an initial state of } A_{\pi(i)}\}$ \\
		The first set indicates that now we want to output value of $x_i$ in the run, which is present in $i^{th}$  block in the input.
		So for that we first move to $\leftend$, and $\Aa_\formula$ remains in state $(q_{\vdash}(x_i), r)$ while performing this operation, i.e. moving toward left end.
		
		\item $ \{(q_{j}(x_i), r) \mid 0 \leq j < i , r \text{~is an initial state of } A_{\pi(i)}\}$ \\
		These states are visited after $(q_{\vdash}(x_i), r)$, specifically when $\Aa_\formula$ is in $(q_{j}(x_i), r)$, it represents that we have already scanned $j$ blocks in the input, and currently reading $(j+1)$th block.

		\item $ \{(q_{i}(x_i), r) \mid r \in Q_{A_{\pi(i)}} \}$ \\
		With the help of above states, $\Aa_\formula$ reaches the desired $i$th block in the input and in the beginning of this block, state is $(q_{i}(x_i), r)$, where $r$ is initial state of $A_{\pi(i)}$. 
		Later $\Aa_\formula$ simulates the transitions of $A_{\pi(i)}$ to check the membership constraint $y_{\pi(i)} \in A_{\pi(i)}$ in the second component of states of $\Aa_\formula$ with the help of third set. 		

	\end{itemize}
	
	\item For each $x_i$, if $x_i$ is a dependent variable, we have following set of states used to produce $x_i$ in the output:
	\begin{itemize}
		\item $\{(q_{\vdash}(x_i, x_{i_j}, j)) \mid \term_{\pi(i)}[j] = x_{i_j}, 1 \leq j \leq \term_{\pi(i)}, p \in Q_{\transducer_{\pi(i)}}, r \in Q_{A_{\pi(i)}} \} $ \\
		The first set of states are similar to above case.
		$\vdash$ present in state name indicates that $\Aa_\formula$ is in process of moving left to reach $\leftend$ and the next, we want to read $x_{i_j}$ from $\term_{\pi(i)}$, position $j$ (i.e. $\term_{\pi(i)}[j] = x_{i_j}$), to produce part of valuation of $x_i$. 
		The reading of value of $x_{i_ j}$ must start from states $p$ in $\transducer_i$, while its output must match with the word starting from $r$ in $A_{\pi(i)}$.
		Here the third parameter $j$ is important since it remembers the progress of $\term_{\pi(i)}$, how many variables from it are already processed.
		
		\item $  \{(q_k(x_i, x_{i_j}, j), (p, r) ) \mid 0 \leq k \leq i_j, \term_{\pi(i)}[j] = x_{i_j}, p \in Q_{\transducer_{\pi(i)}}, r \in Q_{A_{\pi(i)}} \}$ \\
		The second set of states are visited to count the number of $\#$s before reading value of $x_{i_j}$.
		Once this count (presented by a subscript of $q$) reaches $i_j-1$, we are in block $i_j$.
		Then $\Aa_{\formula}$ simulates the transitions of $\transducer_{\pi(i)}$ starting from $p$ and in parallel simulates the transitions of $A_{\pi(i)}$ starting from state $r$ on the output produced by $\transducer_{\pi(i)}$.
			
	\end{itemize}
	
	\item special state $q^{final}$ to mark the end.
\end{itemize}

{\noindent\bf Initial states:}

$$\{(q_0(x_1, x_{1_1}, 1), (p_0, r_0)) \mid x_1 \text{ is dependent }, t_{\pi(1)} = x_{1_1}, p_0, r_0\text{ are initial states of } T_{\pi(1)}, A_{\pi(1)}\}$$

$$\cup \{(q_0(x_1), r_0) \mid x_1 \text{ is independent }, r_0\text{ initial of } A_{\pi(1)}\}$$

{\noindent\bf Final states:}
$\{q^{final}\}$

{\noindent\bf Set of transitions:}
For all $i \in [1, n]$:
\begin{itemize}
	\item If $x_i$ is independent variable, and $r$ is initial state of $A_{\pi(i)}$:
	\begin{enumerate}
		\item 	$(q_{\vdash}(x_i), r) \xrightarrow[\alpha \mid \epsilon, -1]{\# \mid \epsilon, -1} (q_{\vdash}(x_i), r)$ \\
		This loop helps to reach $\leftend$ in the input.
		\item Once we reach the left end $\vdash$ from $(q_{\vdash}(x_i), r)$: \\
		$(q_{\vdash}(x_i), r) \xrightarrow{\vdash \mid \epsilon, +1} (q_{0}(x_i), r)$
		\item To read $x_i$ from input, we need to skip $i$ many $\#$s using following transitions. For all $1 \leq j \leq i-2$: \\
		$(q_{j}(x_i), r) \xrightarrow{\alpha \mid \epsilon, +1} (q_{j}(x_i), r)$ for any $\alpha \in \Sigma$\\
		$(q_{j}(x_i), r) \xrightarrow{\# \mid \epsilon, +1} (q_{j+1}(x_i), r)$ \\
		$(q_{i-1}(x_i), r) \xrightarrow{\alpha \mid \epsilon, 0} (q_{i}(x_i), r)$ for any $\alpha \in \Sigma \cup \{\#, \dashv\}$\\
		These transitions help in counting the number of blocks separated by $\#$. 
		It stops when $\Aa_\formula$ reaches $i$th block.
	\end{enumerate}

	\item 
	Once $\Aa_\formula$ reaches the $i$th block, it needs to simulate the transitions of $A_{\pi(i)}$:
	If $x_i$ is independent variable, and current state is $r$ while checking membership in $A_{\pi(i)}$: \\
	$(q_i(x_i), r) \xrightarrow{\alpha \mid \alpha, +1} (q_i(x_i), r')$ if $(r, \alpha, r') \in \delta_{A_{\pi(i)}}$ , for any $\alpha \in \Sigma$ \\
	Above transition simulates $A_{\pi(i)}$ in second component while producing the same output.
	After it finishes the block $x_i$, it moves ahead to produce $x_{i+1}$ next producing separator $\#$ in the output: 
	\begin{enumerate}
		\item If $x_{i+1}$ is independent variable, and $r_0$ is initial state of $A_{\pi(i+1)}$, $r$ is final state of $A_{\pi(i)}$:\\
		$(q_i(x_i), r) \xrightarrow{\# \mid \#, -1} (q_{\vdash}(x_{i+1}), r_0)$ 
		\item If $x_{i+1}$ is dependent variable, and $r_0$ is initial state of $A_{\pi(i+1)}$, $p_0$ is initial state of $\transducer_{\pi(i+1)}$, $\term_{\pi(i)}[1] = x_{i+1_1}$, $r$ is final state of $A_{\pi(i)}$:\\
		$(q_i(x_i), r) \xrightarrow{\# \mid \#, -1} (q_{\vdash}(x_{i+1}, x_{i+1_1}, 1), (p_0, r_0))$ 		
	\end{enumerate}

	\item If $x_i$ is dependent variable, $r_0$ is initial state of $A_{\pi(i)}$, $p_0$ is initial state of $\transducer_{\pi(i)}$, and $t_{\pi(i)}[1] = x_{i_1}$: \\
	 $(q_{\vdash}(x_{i}, x_{i_1}, 1), (p_0,r_0)) \xrightarrow[\# | \epsilon, -1]{\alpha | \epsilon, -1} (q_{\vdash}(x_{i}, x_{i_1}, 1), (p_0,r_0))$ for any $\alpha \in \Sigma$\\ ~\\
	 This transition moves $\Aa_\formula$ towards $\leftend$. Once $\leftend$ is reached, it starts counting the number of blocks separated by $\#$s. \\ 
	 $(q_{\vdash}(x_{i}, x_{i_1}, 1), (p_0,r_0)) \xrightarrow{\vdash | \epsilon, +1} (q_{0}(x_{i}, x_{i_1}, 1), (p_0,r_0))$ \\
	 
	 \item If $x_i$ is dependent variable, and $t_{\pi(i)}[j] = x_{i_j}$, for any $0 \leq k \leq i_j - 2$: (Counting of $\#$ separated blocks)\\
	 $(q_{k}(x_{i}, x_{i_j}, j), (p,r)) \xrightarrow{\alpha | \epsilon, +1} (q_{k}(x_{i}, x_{i_j}, j), (p,r))$ for any $\alpha \in \Sigma$\\
	 
	 $(q_{k}(x_{i}, x_{i_j}, j), (p,r)) \xrightarrow{\# | \epsilon, +1} (q_{k+1}(x_{i}, x_{i_j}, j), (p,r))$\\
	 
	 \item If $x_i$ is dependent variable, and $t_{\pi(i)}[j] = x_{i_j}$, $k = i_j - 1$ (When the desired block is reached in the input): \\
	 
	 $(q_{k}(x_{i}, x_{i_j}, j), (p,r)) \xrightarrow{\alpha | \epsilon, 0} (q_{k+1}(x_{i}, x_{i_j}, j), (p,r))$ \\
	 
	 \item If $x_i$ is dependent variable, and $t_{\pi(i)}[j] = x_{i_j}$, $k = i_j $(Simulating transitions of $\transducer_{\pi(i)}$ on input, while $A_{\pi(i)}$ on output) : 
	  
	  \begin{enumerate}
	  	\item if $(p, (b, a), p') \in \delta_{T_{\pi(i)}}$ and $(q, b, q') \in \delta_{A_{\pi(i)}}$, then \\
	  	$(q_{k}(x_{i}, x_{i_j}, j), (p,r)) \xrightarrow{a | b, +1} (q_{k}(x_{i}, x_{i_j}, j), (p',r'))$ \\
	  	
	  	\item   if $(p, (b, \epsilon), p') \in \delta_{T_{\pi(i)}}$ and $(q, b, q') \in \delta_{A_{\pi(i)}}$, then \\
	  	$(q_{k}(x_{i}, x_{i_j}, j), (p,r)) \xrightarrow{a | b, 0} (q_{k}(x_{i}, x_{i_j}, j), (p',r'))$ \\
	  	
	  	\item if $|t_{\pi(i)}| > j$, $\term_{\pi(i)}[j+1] = x_{i_{j+1}}$ (If the current variable is not the last in the term, $\Aa_\formula$ proceeds with next variable in term ) \\
	  	$(q_{k}(x_{i}, x_{i_j}, j), (p,r)) \xrightarrow{\# | \epsilon, -1} (q_{\vdash}(x_{i}, x_{i_{j+1}}, j+1), (p,r)) $ \\
	  	
	  	\item if $|\term_{\pi(i)}| = j$, $r$ and $p$ are final states of $A_{\pi(i)}$ and $\transducer_{\pi(i)}$ respectively (If the current variable is the last in the term, $\Aa_\formula$ has finished producing $x_i$ and proceeds with next variable $x_{i+1}$ after producing $\#$ as separator), \\
	  	\begin{enumerate}
	  		\item if $x_{i+1}$ is an independent variable, and $r_0$ is initial state of $A_{\pi(i+1)}$ \\
	  		$(q_{k}(x_{i}, x_{i_j}, j), (p,r)) \xrightarrow{\# | \#, -1} (q_{\vdash}(x_{i+1}), r_0)$ \\
	  		
	  		\item if $x_{i+1}$ is dependent variable, and $r_0$ is initial state of $A_{\pi(i+1)}$, $p_0$ is initial state of $\transducer_{\pi(i+1)}$, $\term_{\pi(i)}[1] = x_{i+1_1}$ \\
	  		$(q_{k}(x_{i}, x_{i_j}, j), (p,r)) \xrightarrow{\# | \#, -1} (q_{\vdash}(x_{i+1}, x_{i+1_1}, 1), (p_0,r_0))$
	  		
	  	\end{enumerate}
	  \end{enumerate}
  
  	\item If $x_n$ is independent variable and $r$ is a final state of $A_{\pi(n)}$ \\
  	$(q_{n}(x_n), r') \xrightarrow[\# | \epsilon, +1]{\dashv | \epsilon, +1} q^{{final}}$ \\
	 
	 \item If $x_n$ is dependent variable, and $t_{\pi(n)}[j] = x_{n_j}$, $k = n_j $, $|\term_{\pi(n)}| = j$, $r$ and $p$ are final states of $A_{\pi(n)}$ and $\transducer_{\pi(n)}$ respectively, \\
	 $(q_{k}(x_{n}, x_{n_j}, j), (p,r)) \xrightarrow[\dashv | \epsilon, +1]{\# | \epsilon, +1} q^{{final}}$ \\
	 
	 \item finally, since the acceptnace condition is to reach final state after reading $\dashv$: \\
	   	$q^{{final}} \xrightarrow[\dashv | \epsilon, +1]{\alpha | \epsilon, +1} q^{{final}}$ 
	 
\end{itemize}

%\end{proof}
\end{document}